\documentclass{amsart}

\usepackage{amsmath}
\usepackage{amsthm}
\usepackage{amssymb}
\usepackage{hyperref}
\usepackage[inline]{enumitem}
\usepackage{cleveref}
\usepackage{xifthen}
\usepackage{xparse}

\usepackage{wrapfig}
\usepackage{cutwin}
\usepackage{tikz}

\usetikzlibrary{automata,backgrounds,calc,positioning}

\newcommand{\cC}{\mathcal{C}}
\newcommand{\cA}{\mathcal{A}}
\newcommand{\cP}{\mathcal{P}}
\newcommand{\cB}{\mathcal{B}}
\newcommand{\cS}{\mathcal{S}}
\newcommand{\cT}{\mathcal{T}}

\newcommand{\cM}{\mathcal{M}}
\newcommand{\fM}{\mathfrak{M}}

\newcommand{\N}{\mathbb{N}}

\newcommand{\eword}{\varepsilon}%
\newcommand{\subword}{\preccurlyeq}%

\DeclareDocumentCommand{\Adh}{O{} m}{\mathsf{Adh}_{#1}(#2)}
\DeclareDocumentCommand{\Dclosure}{O{} m}{\mathord{\downarrow}_{#1}#2}
\DeclareDocumentCommand{\Uclosure}{O{} m}{\mathord{\uparrow}_{#1}#2}

\newcommand{\langof}[1]{L(#1)}
\DeclareDocumentCommand{\seqof}{O{} m}{\mathsf{Runs}_{#1}(#2)}
\newcommand{\transof}[1]{T(#1)}

\newcommand{\autstep}[1][YYY]{\ifthenelse{\equal{#1}{YYY}}{\rightarrow}{\rightarrow_{#1}}}
\newcommand{\autsteps}[1][YYY]{\ifthenelse{\equal{#1}{YYY}}{\xrightarrow{*}}{\xrightarrow{*}_{#1}}}

\DeclareDocumentCommand{\MOD}{O{}}{\mathsf{mod}_{#1}}
\DeclareDocumentCommand{\REG}{}{\mathsf{reg}}
\DeclareDocumentCommand{\LTT}{O{}}{\mathsf{LTT}_{#1}}
\DeclareDocumentCommand{\BSS}{O{}}{\ifthenelse{\isempty{#1}}{\mathcal{B}\Sigma_1[\mathord{<}]}{\mathcal{B}\Sigma_1[\mathord{<},#1]}}

\newcommand{\gord}{\preceq}
\newcommand{\ford}[1]{\le_{#1}}
\newcommand{\embedord}{\hookrightarrow}
\newcommand{\autord}[1]{\preceq_{#1}}

\newcommand{\modord}[1]{\preceq_{#1}}
\newcommand{\cord}[1]{\preceq_{#1}}
\newcommand{\morord}[1]{\preceq_{#1}}

\newcommand{\DesImp}[2]{``\labelcref{#1}$\Rightarrow$\labelcref{#2}''}

\newcommand{\Powerset}[1]{\mathcal{P}(#1)}

\theoremstyle{plain}                                                                           
\newtheorem{thm}{Theorem}[section]                                                             
\newtheorem{prop}[thm]{Proposition}                                                            
\newtheorem{lem}[thm]{Lemma}                                                                   
\newtheorem{cor}[thm]{Corollary}                                                               
\newtheorem{obs}[thm]{Observation}                                                               

\newtheorem{prule}{Rule}

\theoremstyle{remark}
\newtheorem{rem}[thm]{Remark}

\crefname{thm}{Theorem}{Theorems}                                                              
\Crefname{thm}{Theorem}{Theorems}                                                              
\crefname{prop}{Proposition}{Propositions}                                                     
\Crefname{prop}{Proposition}{Propositions}                                                     
\crefname{lem}{Lemma}{Lemmas}                                                                  
\Crefname{lem}{Lemma}{Lemmas}                                                                  
\crefname{cor}{Corollary}{Corollaries}                                                         
\Crefname{cor}{Corollary}{Corollaries}  
\crefname{rem}{Remark}{Remarks}                                                         
\Crefname{rem}{Remark}{Remarks}  
\crefname{ex}{Example}{Examples}                                                         
\Crefname{ex}{Example}{Examples}  
\crefname{prule}{Rule}{Rules}                                                         
\Crefname{prule}{Rule}{Rules}  
\crefname{obs}{Observation}{Observations}                                                         
\Crefname{obs}{Observation}{Observations}

\begin{document}

\title{PTL-separability and closures for WQOs on words} 

\author{Georg Zetzsche}
\address{IRIF (Uniersit\'{e} Paris-Diderot, CNRS), France, \texttt{zetzsche@irif.fr}}
\thanks{Supported by a fellowship of the Fondation Sciences Math\'{e}matiques de Paris.}

\begin{abstract}
  We introduce a flexible class of well-quasi-orderings (WQOs) on
  words that generalizes the ordering of (not necessarily contiguous)
  subwords.
  Each such WQO induces a class of piecewise testable languages (PTLs)
  as Boolean combinations of upward closed sets. In this way, a range
  of regular language classes arises as PTLs. Moreover, each of the
  WQOs guarantees regularity of all downward closed sets.
  We consider two problems. First, we study which (perhaps
  non-regular) language classes permit a decision procedure to decide
  whether two given languages are separable by a PTL with respect to a
  given WQO. Second, we want to effectively compute downward closures
  with respect to these WQOs.
  Our first main result that for each of the WQOs, under mild
  assumptions, both problems reduce to the simultaneous unboundedness
  problem (SUP) and are thus solvable for many powerful system
  classes. In the second main result, we apply the framework to show
  decidability of separability of regular languages by
  $\mathcal{B}\Sigma_1[<, \mathsf{mod}]$, a fragment of
  first-order logic with modular predicates.
\end{abstract}

\maketitle

\section{Introduction}
In the verification of infinite-state systems, it is often useful to
construct finite-state abstractions.  This is because finite-state
systems are much more amenable to analysis.  For example, if a
pertinent property of our system is reflected in a finite-state
abstraction, then we can work with the abstraction instead of the
infinite-state system itself. Another example is that the abstraction
acts as a certificate for correctness: A violation free
overapproximation of the set of behaviors of a system certifies
absence of violations in the system itself.
Here, we study two types of such abstractions: \emph{downward
  closures}, which are overapproximations of individual languages and
\emph{separators} as certificates of disjointness.

\subsection*{Downward closures} A particularly appealing abstraction
is the \emph{downward closure}, the set of all (not necessarily
contiguous) subwords of the members of a language. What makes this
abstraction interesting is that since the subword ordering is a
well-quasi-ordering (WQO), the downward closure of \emph{any} language
is regular~\cite{Higman1952,Haines1969}. Recently, there has been
progress on when the downward closure is not only regular but can also
be effectively computed. It is known that downward closures are
computable for \emph{context-free
  languages}~\cite{Courcelle1991,vanLeeuwen1978}, \emph{Petri net
  languages}~\cite{HabermehlMeyerWimmel2010}, and \emph{stacked
  counter automata}~\cite{Zetzsche2015a}. Moreover, recently, a
general sufficient condition for computability
 was presented in
\cite{Zetzsche2015b}. Using the latter, downward closures were then
shown to be computable for \emph{higher-order pushdown
  automata}~\cite{HagueKochemsOng2016} and \emph{higher-order
  recursion
  schemes}~\cite{ClementeParysSalvatiWalukiewicz2016}. Hence, downward
closures are computable for very powerful models.

If we want to use downward closures to prove absence of violations,
then using the downward closure in this way has the disadvantage that
it is not obvious how to refine it, i.e. systematically construct a
more precise overapproximation in case the current one does not
certify absence of violations. Therefore, we wish to find abstractions
that are refinable in a flexible way and still guarantee regularity
and computability.

\subsection*{Separability}
Another type of finite-state abstractions is that of separators.
Since safety properties of multi-threaded programs can often be
formulated as the disjointness of two languages, one
approach to this task is to use regular languages to certify
disjointness~\cite{BouajjaniEsparzaTouili2003,ChakiEtAl2006,LongCalinMajumdarMeyer2012}.
A \emph{separator} of two languages $K$ and $L$ is a set $S$ such that
$K\subseteq S$ and $L\cap S=\emptyset$.  Therefore, especially in
cases where disjointness of languages is undecidable or hard, it would
be useful to have a decision procedure for the \emph{separability
  problem}: Given two languages, it asks whether they are separable by
a language from a particular class of separators.  In particular, if
we want to apply such algorithms to infinite-state systems, it would
be desirable to find large classes of separators (and systems) for
which the separability problem is decidable.

It has long been known that separability of context-free languages are
undecidable already for very simple classes of regular
languages~\cite{SzymanskiWilliams1976,Hunt1982} and this stifled hope
that separability would be decidable for any interesting classes of
infinite-state systems and classes of separators. However, the subword
ordering turned out again to have excellent decidability properties:
It was shown recently that for a wide range of language classes, it is
decidable whether two given languages are separable by a piecewise
testable language (PTL)~\cite{CzerwinskiMartensRooijenZeitoun2015}.  A
PTL is a finite Boolean combination of upward closures (with respect
to the subword ordering) of single words.  In fact, in turned out that
(under mild closure assumptions) separability by PTL is decidable if
and only if downward closures are
computable~\cite{CzerwinskiMartensRooijenZeitounZetzsche2017a}.

However, while this separability result applies to very expressive
models of infinite-state systems, it is still limited in terms of the
separators: The small class of PTL will not always suffice as
disjointness certificates. 

\subsection*{Contribution} This work makes two contributions, a
conceptual one and a technical one.  The conceptual contribution is
the introduction of a fairly flexible class of WQOs on words. These
are refinable and provide generalizations of the subword
ordering. These orders are parameterized by transducers, counter
automata or other objects and can be chosen to reflect various
properties of words. Moreover, the classes of corresponding PTLs are a
surprisingly rich collection of classes of regular languages.

Moreover, it is shown that all these orders have the same pleasant
properties in terms of downward closure computation and decidability
of PTL-separability as the subword ordering. More specifically, it is
shown that (under mild assumptions), decidability of the
abovementioned unboundedness problem again
characterizes \begin{enumerate*} \item those language classes for
  which downward closures are computable and \item those classes where
  separability by PTL is decidable.\end{enumerate*}

In addition, it turns out that this framework can also be used to
obtain decidable separability of regular languages by $\BSS[\MOD]$, a
fragment of first-order logic with modular predicates. This is
technically relatively involved and generalizes the fact that
definability of regular languages in
$\BSS[\MOD]$ is decidable~\cite{ChaubardPinStraubing2006}.
\section{Preliminaries}
If $\Sigma$ is an alphabet, $\Sigma^*$ denotes the set of words over
$\Sigma$.  The empty word is denoted by $\eword\in \Sigma^*$.  A
\emph{quasi-order} is an ordering that is reflexive and transitive.
An ordering $(X,\gord)$ is called a \emph{well-quasi-ordering (WQO)}
if for every sequence $x_1,x_2,\ldots\in X$, there are indices $i<j$
with $x_i\gord x_j$. This is equivalent to requiring that every
sequence $x_1,x_2,\ldots\in X$ contains an infinite subsequence $x'_1,x'_2,\ldots\in X$ that is \emph{ascending}, meaning $x'_i\gord x'_j$ for $i\le j$. For
a subset $L\subseteq X$, we define $\Dclosure[\gord]{L}=\{x\in X \mid
\exists y\in L\colon x\gord y\}$ and $\Uclosure[\gord]{L}=\{x\in X
\mid \exists y\in L\colon y\gord x\}$.  These are called the
\emph{downward closure} and \emph{upward closure} of $L$,
respectively.  A set $L\subseteq X$ is called \emph{downward closed}
(\emph{upward closed}) if $\Dclosure[\gord]{L}=L$
($\Uclosure[\gord]{L}=L$). A (defining) property of
well-quasi-orderings is that for every non-empty upward-closed set
$U$, there are finitely many elements $x_1,\ldots,x_n\in U$ such that
$U=\Uclosure[\gord]{\{x_1,\ldots,x_n\}}$.  See~\cite{Kruskal1972} for
an introduction. An ordering $(\Sigma^*,\gord)$ on words is called
\emph{multiplicative} if $u_1\gord v_1$ and $u_2\gord v_2$ implies
$u_1u_2\gord v_1v_2$.

For words $u,v\in\Sigma^*$, we write $u\subword v$ if $u=u_1\cdots
u_n$ and $v=v_0u_1 v_1\cdots u_nv_n$ for some
$u_1,\ldots,u_n,v_0,\ldots,v_n\in\Sigma^*$. This ordering is called
the \emph{subword ordering} and it is well-known that this is a
well-quasi-ordering~\cite{Higman1952}.  

A well-studied class of regular languages is that of the piecewise
testable languages. Classically, a language $L\subseteq\Sigma^*$ is a
\emph{piecewise testable language (PTL)}~\cite{Simon1975} if it is a
finite Boolean combination of sets of the form
$\Uclosure[\subword]{w}$ for $w\in\Sigma^*$.  However, this notion
makes sense for any WQO $(X,\gord)$~\cite{GoubaultLarrecqSchmitz2016}
and we call a set $L\subseteq X$ \emph{piecewise testable} if it is a
finite Boolean combination of sets $\Uclosure[\gord]{x}$ for $x\in X$.

A \emph{(finite-state)
  transducer} is a finite automaton where every edge reads input and
produces output.  For a transducer $T$ and a language $L$, the
language $T L$ consists of all words output by the transducer while
reading a word from $L$.  A class of languages $\cC$ is called a
\emph{full trio} if it is effectively closed under rational
transductions, i.e. if $TL\in\cC$ for each $L\in\cC$ and each rational
transduction $T$.

\section{Parameterized WQOs and main results}\label{results}
In this \lcnamecref{results}, we introduce the parameterized WQOs on
words, state the main results of this work, and present some
applications. We define the class of parameterized WQOs inductively
using rules
(\cref{prule:subword,prule:transductions,prule:conjunctions}). The
simplest example is Higman's subword ordering.

\begin{prule}\label{prule:subword}
For each $\Sigma$, $(\Sigma^*,\subword)$ is a parameterized WQO.
\end{prule}

\subsection*{Orderings defined by transducers} To make things more
interesting, we have a type of WQOs that are defined by functions.
Suppose $X$ and $Y$ are sets and we have a function $f\colon X\to Y$.
A general way of constructing a WQO on $X$ is to take a WQO
$(Y,\gord)$ and set $x\gord_f x'$ if and only if $f(x)\gord f(x')$. It
is immediate from the definition that then $\gord_f$ is a WQO on
$X$. We apply this idea to transducers.

 A \emph{finite-state transducer} over $\Sigma$ and $\Gamma$ is a
tuple $\cT=(Q,\Sigma,\Gamma,E,I,F)$, where $Q$ is a finite set of states,
$E\subseteq Q\times
(\Sigma\cup\{\eword\})\times(\Gamma\cup\{\eword\})\times Q$ is its set
of \emph{edges}, $I\subseteq Q$ is the set of \emph{initial states}, and
$F\subseteq Q$ is the set of \emph{final states}. Transducers accept sets of
pairs of words. A \emph{run} of $\cT$ is a sequence
\[ (q_0,u_1,v_1,q_1)(q_1,u_2,v_2,q_2)\cdots (q_{n-1},u_n,v_n,q_n) \]
 of
edges such that $q_0\in I$, $q_n\in F$. The pair \emph{read by the run}
is $(u_1\cdots u_n,v_1\cdots v_n)$. Then, $\cT$ \emph{realizes} the relation
\[ \transof{\cT}=\{(u,v)\in\Sigma^*\times\Gamma^* \mid \text{$(u,v)$ is read by a run of $\cT$}\}. \]
Relations of this form are called \emph{rational transductions}.  A
transduction is \emph{functional} if for every $u\in\Sigma^*$,
there is \emph{exactly} one $v\in\Gamma^*$ with $(u,v)\in
\transof{\cT}$. In other words, $\transof{\cT}$ is a function
$\transof{\cT}\colon\Sigma^*\to\Gamma^*$ and we can use it to define a
WQO.

\begin{prule}\label{prule:transductions}
  Let $f\colon \Sigma^*\to\Gamma^*$ be a functional transduction.
  If $(\Gamma^*,\gord)$ is a parameterized WQO, then so is
  $(\Sigma^*,\gord_f)$.
\end{prule}

\subsection*{Conjunctions} Another way to build a WQO on a set is to
combine two existing WQOs. Suppose $(X,\gord_1)$ and $(X,\gord_2)$ are
WQOs. Their \emph{conjunction} is the ordering $(X,\gord)$ with
$x\gord x'$ if and only if $x\gord_1 x'$ and $x\gord_2 x'$. Then
$(X,\gord)$ is a WQO via the characterization using ascending
subsequences.

\begin{prule}\label{prule:conjunctions}
  If $(\Sigma^*,\gord_1)$ and $(\Sigma^*,\gord_2)$ are parameterized
  WQOs, then so is their conjunction $(\Sigma^*,\gord)$.
\end{prule}

\subsection*{Examples}
Using the three building blocks in
\cref{prule:subword,prule:transductions,prule:conjunctions}, we can
construct a wealth of WQOs on words. Let us mention a few examples,
including the accompanying classes of PTL.

\subsection*{Labeling transductions} Our first class of examples
concerns orderings whose PTLs are fragments of first-order logic with
additional predicates. A \emph{labeling transduction} is a functional
transduction $f\colon \Sigma^*\to(\Sigma\times\Lambda)^*$ for some
alphabet $\Lambda$ labels such that for each $w=a_1\cdots
a_n\in\Sigma^*$, $a_1,\ldots,a_n\in\Sigma$, we have
$f(w)=(a_1,\ell_1)\cdots(a_n,\ell_n)$ for some
$\ell_1,\ldots,\ell_n\in\Lambda$.

In this case, we can interpret $\subword_f$-PTLs logically. To each
word $w=a_1\cdots a_n$, $a_1,\ldots,a_n\in\Sigma$, we associate a
finite relational structure $\fM_{w}$ as follows. Its domain is
$D=\{1,\ldots,n\}$ and as predicates, it has the binary $<$, unary
letter predicates $P_a$ for $a\in\Sigma$, and for each $\ell\in
\Lambda$, we have a unary predicate $\pi_\ell$.
While the predicates $<$ and $P_a$ are interpreted as expected, we
have to explain $\pi_\ell$. If $f(w)=(a_1,\ell_1)\cdots (a_n,\ell_n)$,
then $\pi_\ell(i)$ expresses that $\ell_i=\ell$. Hence, the
$\pi_\ell$ give access to the labels produced by $f$.  We denote the
$\cB\Sigma_1$-fragment (Boolean combinations of $\Sigma_1$-formulas)
as $\BSS[f]$.

Suppose $\fM_1$ and $\fM_2$ are relational structures over the same
signature.  An \emph{embedding} of $\fM_1$ in $\fM_2$ is an injective
mapping from the domain of $\fM_1$ to the domain of $\fM_2$ such that
each predicate holds for a tuple in $\fM_1$ if and only the predicate
holds for the image of that tuple.  This defines a quasi-ordering: We
write $\fM_1\embedord\fM_2$ if $\fM_1$ can be embedded into $\fM_2$.
Observe that for $u,v\in\Sigma^*$, we have $u\subword_f v$ if and only
if $\fM_u\embedord\fM_v$.

It was shown in \cite{GoubaultLarrecqSchmitz2016} that if the embedding order
is a WQO on a set of structures, then the
$\cB\Sigma_1$-fragment (i.e. Boolean combinations of $\Sigma_1$ formulas) can
express precisely the PTL with respect to $\embedord$. This implies that the
languages definable in $\BSS[f]$ are precisely the $\subword_f$-PTL.

To illustrate the utility of the fragments $\BSS[f]$, suppose we are
given regular languages $W_i$, $P_i$, $S_i$, for $i\in[1,n]$.  Suppose
we have for each $i\in[1,n]$ a $0$-ary predicate $\mathsf{w}_i$ that
expresses that our whole word belongs to $W_i$. For each $i\in[1,n]$
we also have unary predicates $\mathsf{pre}_i$ and $\mathsf{suf}_i$,
which express that the prefix and suffix, respectively, corresponding
to the current position, belongs to $P_i$ and $S_i$, respectively.
Then the corresponding fragment 
\[ \BSS[(\mathsf{w}_i)_{i\in[1,n]},(\mathsf{pre}_i)_{i\in[1,n]}, (\mathsf{suf}_i)_{i\in[1,n]}]\]
can clearly be realized as $\BSS[f]$.

Of course, we can capture many other predicates by labeling
transducers. For example, it is easy to realize a predicates for ``the
distance to the closest position to the left with an $a$ is congruent
$k$ modulo $d$'' (for some fixed $d$).  Finally, let us observe in passing that
instead of enriching $\BSS$, we could also construct fragments that do
not have access to letters: If $f$ just produces labels (and no input
letters), we obtain a logic where, for example, we can only express
whether ``this position is even \emph{and} carries an $a$''.

\subsection*{Orderings defined by finite automata} Our second
example slightly specializes the first example. The reason we make it
explicit is that we shall present explicit ideal representations that
will be applied to decide separability of regular languages by
$\BSS[\MOD]$. The example still generalizes the subword order. While
in the latter, a smaller word is obtained by deleting arbitrary
infixes, these orders use an automaton to restrict the permitted deletion.

A \emph{finite automaton} is a tuple $\cA=(Q,\Sigma,E,I,F)$, where $Q$ is a
finite set of \emph{states}, $\Sigma$ is the \emph{input alphabet}, $E\subseteq
Q\times \Sigma\times Q$ is the set of \emph{edges}, $I\subseteq Q$ is the set of
\emph{initial states}, and $F\subseteq Q$ is the set of \emph{final states}.
The language $\langof{\cA}$ is defined in the usual way.
Here, we use automata as a means to assign a labeling to an input word. A
labeling is defined by a run. A \emph{run} of $\cA$ on $w=a_1\cdots a_n$,
$a_1,\ldots,a_n\in\Sigma$, is a sequence 
\[ (q_0,a_1,q_1)(q_1,a_2,q_2)\cdots
(q_{n-1},a_n,q_n)\in E^* \]
 with $q_0\in I$ and $q_n\in F$. By $\seqof{\cA}$,
denote the set of runs of $\cA$.  Since we want $\cA$ to label every word from $\Sigma^*$, we
call an automaton $\cA$ a \emph{labeling automaton} if for each word $w\in\langof{\cA}$,
$\cA$ has exactly one run on $w$. In this case, we write $\cA(w)$ for the run of
$\cA$ on $w$. Moreover, we define $\sigma_{\cA}(w)=(p,q)$, where $p$ and $q$
are the first and last state, respectively, visited during $w$'s run.
Hence, such an automaton defines a map $\cA\colon \Sigma^*\to E^*$.
Let $u\autord{\cA} v$ if and only if $v$ is obtained from $u$ by ``inserting
loops of $\cA$''.  In other words, $v$ can be written as $v=u_0v_1u_1\cdots v_nu_n$ with
$u=u_0\cdots u_n$ such that the run of $\cA$ on $v$ occupies the same state
before reading $v_i$ and after reading $v_i$. Equivalently, we have
$u\autord{\cA} v$ if and only if $\sigma_{\cA}(u)=\sigma_{\cA}(v)$ and
$\cA(u)\subword\cA(v)$.  The order $\autord{\cA}$ is a parameterized WQO: 
The order $\gord$ with $u\gord v$ if and only if $\sigma_{\cA}(u)=\sigma_{\cA}(v)$
is parameterized because we can use a functional transduction $f$ that maps 
$u$ to the length-1 word $\sigma_{\cA}(u)$ in $(Q\times Q)^*$. Moreover,
with a functional transduction $g$ that maps a word $w$ to its run $\cA(w)$, 
the ordering $\autord{\cA}$ is the conjunction of $\subword_f$ and $\subword_g$.
\begin{itemize}
\item If $\cA$ consists of just one state and a loop for every $a\in \Sigma$,
then $\autord{\cA}$ is the ordinary subword ordering.
\item Suppose $\cB$ is a complete deterministic automaton accepting a regular language
$L\subseteq\Sigma^*$. Then $L$ is simultaneously upward closed and downward
closed with respect to $\autord{\cA}$, where $\cA$ is obtained from $\cB$ by
making all states final.
In particular, every regular language
can occur as an upward closure and as a downward closure with respect to
some $\autord{\cA}$.
\end{itemize}

As for labeling transducers, we can consider logical fragments where
$\autord{\cA}$ is the embedding order. Again, our signature consists
of $<$, $P_a$ for $a\in\Sigma$. Furthermore, for each $q\in Q$, we
have the $0$-ary predicates $\iota_q$ and $\tau_q$ and unary
predicates $\lambda_q$ and $\rho_q$.  Let $(q_0,a_1,q_1)\cdots
(q_{n-1},a_n,q_n)$ be the run of $\cA$ on $w$.  Then $\lambda_q(i)$ is
true iff $q_{i-1}=q$. Moreover, $\rho_q(i)$ holds iff $q_{i}=q$.
Hence, $\lambda_q$ and $\rho_q$ give access to the state occupied by
$\cA$ to the left and to the right of each position,
respectively. Accordingly, $\iota_q$ and $\tau_q$ concern the first
and the last state: $\iota_q$ is satisfied iff $q_0=q$ and $\tau_q$ is
true iff $q_n=q$.

\label{def-md}
As an example, let $\cM_d$ be the automaton that consists of a single
cycle of length $d$ so that on each input letter, $\cM_d$ moves one
step forward in the cycle. This is equivalent to having a predicate
for each $k\in[1,d]$ that express that the current position is
congruent $k$ modulo $d$. Moreover, we have a predicate for each
$k\in[1,d]$ to express that the length of the word is $k$ modulo
$d$. This is sometimes denoted $\BSS[\MOD[d]]$.  If these predicates
are available for \emph{every} $d$, the resulting class is denoted
$\BSS[\MOD]$~\cite{ChaubardPinStraubing2006} and will be the subject
of \cref{separability:mod:decidable}.

\subsection*{Multiplicative well-partial orders}
Ehrenfeucht~et~al.~\cite{EhrenfeuchtHausslerRozenberg1983} have shown that a
language is regular if and only if it is upward closed with respect to some
multiplicative WQO. For the ``only if'' direction, they provide the syntactic
congruence, which, as a finite-index equivalence, is a WQO%
. Here, we exhibit a natural example of a well-\emph{partial} order 
for which a given regular language is upward closed.
Suppose $M$ is a finite monoid and $\theta\colon \Sigma^*\to M$ is a
morphism that recognizes the language $L\subseteq\Sigma^*$, i.e.
$L=\theta^{-1}(\theta(L))$. Let $f\colon \Sigma^*\to (M^2\times\Sigma\times M^2)^*$
be the functional transduction such that for $w=a_1\cdots a_n$, $a_1,\ldots,a_n\in\Sigma$,
we have $f(w)=(\ell_0,r_0,a_1,\ell_1,r_1)\cdots (\ell_{n-1},r_{n-1},a_n,\ell_n,r_n)$, where 
$\ell_i=\theta(a_1\cdots a_i)$ and $r_i=\theta(a_{i+1}\cdots a_n)$.
Then we have
$u\subword_f v$ if and only if $v$ can be written as $v=u_0v_1u_1\cdots
v_nu_n$ such that $\theta(u_0\cdots u_{i-1}v_i)=\theta(u_0\cdots u_{i-1})$ and
$\theta(v_iu_i\cdots u_n)=\theta(u_i\cdots u_n)$ for $i\in[1,n]$.
In this case, we write $\morord{\theta}$ for $\subword_f$.

Note that $\morord{\theta}$ is multiplicative and $L$ is
$\morord{\theta}$-upward closed.  Thus, the order $\morord{\theta}$ is
a natural example that shows: A language is regular if and only if it
is upward closed with respect to some multiplicative well-partial
order.
\begin{rem}
Another source of WQOs on words is \cite{BucherEhrenfeuchtHaussler1985}, where
Bucher~et~al. have studied a class of
multiplicative orderings on words arising from rewriting systems.  They show
that all WQOs considered there can be represented by finite monoids equipped
with a multiplicative quasi-order.  Given such a monoid $(M,\le)$ and a
morphism $\theta\colon \Sigma^*\to M$, they set $u \sqsubseteq_\theta v$ if
and only if $u=u_1\cdots u_n$, $u_1,\ldots,u_n\in\Sigma$, and $v=v_1\cdots v_n$
such that $\theta(u_i)\le\theta(v_i)$. However, they leave open for which
monoids $(M,\le)$ the order $\sqsubseteq_\theta$ is a WQO.

In the case that $\theta$ above is a morphism into a finite group (whose order
is the equality), the order $\morord{\theta}$ coincides with
$\sqsubseteq_{\theta}$. However, while the orderings considered by
Bucher~et~al. are always multiplicative, this is
not always the case for parameterized WQOs.
\end{rem}

\subsection*{Orderings defined by counter automata}
We can also use automata with counters to produce parameterized WQOs.  A \emph{counter
automaton} is a tuple $\cA=(Q,\Sigma,C,E,I,F)$, where $Q$ is a finite set of
\emph{states}, $\Sigma$ is the input alphabet, $C$ is a set of \emph{counters},
$ E\subseteq Q\times (A\cup \{\eword\}) \times \N^C \times Q$ is the finite set
of \emph{edges}, $I\subseteq Q$ is the set of \emph{initial states}, and
$F\subseteq Q$ is the set of \emph{final states}.  A \emph{configuration of
$\cA$} is a tuple $(q,w,\mu)$, where $q\in Q$, $w\in A^*$, $\mu\in \N^C$.
The step relation is defined as follows. We have $(q,w,\mu)\autstep[\cA]
(q',w',\mu')$ iff there is an edge $(q,v,\nu,q')\in E$ such that $w'=wv$ and
$\mu'=\mu+\nu$. A \emph{run (arriving at $\mu$)} on an input word $w$ is a sequence
$(q_0,w_0,\mu_0),\ldots,(q_n,w_n,\mu_n)$ such that
$(q_{i-1},w_{i-1},\mu_{i-1})\autstep[\cA](q_i,w_i,\mu_i)$ for $i\in[1,n]$,
$q_0\in I$, $w_0=\eword$, $\mu_0=0$, $q_n\in F$, and $w_n=w$. 

We use counter automata not primarily as accepting devices, but rather
to define maps and to specify unboundedness properties.  We call $\cA$
a \emph{counting automaton} if it has exactly one run for every word
$w\in\Sigma^*$.  In this case, it defines a function
$\cA\colon\Sigma^*\to\N^C$: We have $\cA(w)=\mu$ iff $\cA$ has a run
on $w$ arriving at $\mu$.

This gives rise to an ordering: Let $\cA$ be a counting
automaton. Then, given $u,v\in\Sigma^*$, let $u \cord{\cA} v$ if and
only if $\cA(u)\le\cA(v)$.  This is a parameterized WQO for the
following reason. For each $c\in C$, we can build a functional
transduction $f_c\colon \Sigma^*\to\{c\}^*$ that operates like $\cA$,
but instead of incrementing $c$, it outputs a $c$.  Then, $\cord{\cA}$
is the conjunction of all the WQOs $\subword_{f_c}$ for $c\in C$.

Let $k\in\N$ and $C_k=\{a_u, b_u, c_u \mid u\in\Sigma^{\le k}\}$.
We say that a word $u$ \emph{occurs at position $\ell$ in $v$} if $v=xuy$
with $|x|=\ell-1$.
It is easy to construct a counting automaton $\cP_k$ with counter set
$C_k$ that satisfies $\cP_k(w)=\mu$ iff for each $u\in\Sigma^{\le k}$, 
\begin{itemize}[leftmargin=0.5cm]
\item if $u$ is a prefix of $w$, then $\mu(a_u)=1$, otherwise $\mu(a_u)=0$,
\item if $u$ is a suffix of $w$, then $\mu(b_u)=1$, otherwise $\mu(b_u)=0$,
\item $\mu(c_u)$ is the number of positions in $w$ where $u$ occurs.
\end{itemize}
Using this counting automaton, we can realize another class of regular languages. Let
$k\in\N$. A \emph{$k$-locally threshold testable language} is a finite Boolean
combination of sets of the form
\begin{itemize}
\item $u\Sigma^*$ for some $u\in\Sigma^{\le k}$,
\item $\Sigma^*u$ for some $u\in\Sigma^{\le k}$, or
\item $\{w\in\Sigma^* \mid \text{$u$ occurs at $\ge \ell$ positions in $w$}\}$ for some $u\in\Sigma^{\le k}$ and $\ell\in\N$.
\end{itemize}
The class of $k$-locally threshold testable languages is denoted
$\LTT[k]$.  Observe that the $\cord{\cP_k}$-PTL are precisely the
$k$-locally threshold testable languages.  Indeed, each of the basic
building blocks of $k$-locally threshold testable languages is
$\cord{\cP_k}$-upward closed and hence a
$\cord{\cP_k}$-PTL. Conversely, for each $w\in\Sigma^*$, the upward
closure of $w$ with respect to $\cord{\cP_k}$ is clearly in $\LTT[k]$.

\subsection*{Conjunctions} Let us illustrate the utility of
conjunctions.  Let $S$ be a finite collection of WQOs on
$\Sigma^*$. We call a language $L\subseteq\Sigma^*$ an \emph{$S$-PTL}
if it is a finite Boolean combination of sets of the form
$\Uclosure[\gord]{w}$, where $\gord$ belongs to $S$ and
$w\in\Sigma^*$. Our framework also applies to $S$-PTLs for the following reason.
\begin{obs}\label{boolean-conjunctions}
Let $\gord$ be the conjunction of the WQOs in $S$.
Then a language is an $S$-PTL iff it is a $\gord$-PTL.
\end{obs}
As an example, suppose we have subsets
$\Sigma_1,\ldots,\Sigma_n\subseteq\Sigma$ and the functional
transductions $\pi_i$, $i\in[1,n]$, such that
$\pi_i\colon\Sigma^*\to\Sigma_i^*$ is the projection onto $\Sigma_i$,
meaning $\pi_i(a)=a$ for $a\in\Sigma_i$ and $\pi_i(a)=\eword$ for
$a\notin\Sigma_i$. If $S$ consists of the $\subword_{\pi_i}$ for
$i\in[1,n]$, then the $S$-PTL are precisely those languages that are
Boolean combinations of sets $\Uclosure[\subword]{w}$ for
$w\in\Sigma_1^*\cup\cdots\cup\Sigma_i^*$. Hence, we obtain a subclass
of the classical PTL. Of course, there are many other examples. One
can, for example, combine WQOs for logical fragments with WQOs defined
by counting automata and thus obtain logics that refer to positions as
well as counter values, etc.

\subsection*{Computing downward closures} 
The first problem we will study is that of computing downward closures.
As in the case of the subword ordering, we will see that for all
parameterized WQOs, every downward closed language is regular.
While mere regularity is often easy to see, it is not obvious how,
given a language $L\subseteq\Sigma^*$, to compute a finite automaton
for $\Dclosure[\gord]{L}$. We are insterested in when this can be done
algorithmically.
If $\gord$ is a WQO on words, we say that \emph{$\gord$-downward closures are
computable} for a language class $\cC$ if there is an algorithm that, given a
language $L\subseteq\Sigma^*$ from $\cC$, computes a finite automaton for
$\Dclosure[\gord]{L}$. This is especially interesting when $\cC$ is a class
of languages of infinite-state systems. 

Until now, downward closure computation has focused mainly on the case
where $\gord$ is the subword ordering. In that case, there is a
charaterization for when downward closures are
computable~\cite{Zetzsche2015b}. For a rational transduction
$T\subseteq\Sigma^*\times\Gamma^*$ and a language
$L\subseteq\Sigma^*$, let $TL=\{v\in\Gamma^*\mid \exists u\in L\colon
(u,v)\in T\}$.  When we talk about \emph{language classes}, we always
assume that there is a way of representing their languages such as by
automata or grammars.  We call a language class $\cC$ a \emph{full
  trio} if it is effectively closed under rational transductions,
i.e. given a representation of $L$ from $\cC$, we can compute a
representation of $TL$ in $\cC$.  The \emph{simultaneous unboundedness
  problem (SUP)} for $\cC$ is the following decision problem.
\begin{description}
\item[Given] A language $L\subseteq a_1^*\cdots a_n^*$ from $\cC$. 
\item[Question] Does $a_1^*\cdots a_n^*\subseteq \Dclosure{L}$ hold?
\end{description}
The aforementioned characterization now states that downward closures
for the subword ordering are computable for a full trio $\cC$ if and
only if the SUP is decidable.  The SUP is decidable for many important
and very powerful infinite-state systems.  It is known to be decidable for
Petri net
languages~\cite{CzerwinskiMartensRooijenZeitounZetzsche2017a,Zetzsche2015b,HabermehlMeyerWimmel2010}
and matrix languages~\cite{Zetzsche2015b}. Moreover, it was shown to be
decidable for indexed languages~\cite{Zetzsche2015b}, which was
generalized to higher-order pushdown
automata~\cite{HagueKochemsOng2016} and then further to higher-order
recursion schemes~\cite{ClementeParysSalvatiWalukiewicz2016}.

An indication for why computing downward closures for parameterized
WQOs might be more difficult than for subwords is that the latter
ordering is a rational relation,
i.e. $\{(u,v)\in\Sigma^*\times\Sigma^*\mid u\subword v\}$ is
rational. This fact was crucial for the method in
\cite{Zetzsche2015b}. However, one can easily construct parameterized
WQOs for which this is not the case.

\subsection*{PTL and separability} We also consider separability
problems. We say that two languages $K\subseteq\Sigma^*$ and
$L\subseteq\Sigma^*$ are \emph{separated} by a language
$R\subseteq\Sigma^*$ if $K\subseteq R$ and $L\cap R=\emptyset$.  If
two languages are separated by a regular language, we can regard this
regular language as a finite-state abstraction of the two
languages. We therefore want to decide when two given languages can be
separated by a language from some class of separators. More precisely,
we say that for a language class $\cC$ and a class of separators
$\cS$, \emph{separability by $\cS$ is decidable} if given language $K$
and $L$ from $\cC$, it is decidable whether there is an $R$ in $\cS$
that separates $K$ and $L$. In the case where $\cS$ is the class
(subword) PTL, it is known when separability is decidable:
In~\cite{CzerwinskiMartensRooijenZeitounZetzsche2017a}, it was shown
that in a full trio, separability by PTL is decidable if and only if
the SUP is decidable (the ``if'' direction had been obtained
in~\cite{CzerwinskiMartensRooijenZeitoun2015}).

\subsection*{Main result}
We are now ready to state the first main result.
\begin{thm}\label{equivalence}
For every full trio $\cC$, the following are equivalent:
\begin{enumerate}
\item\label{eq:sup} The SUP is decidable for $\cC$.
\item\label{eq:dc} For every parameterized WQO $\gord$, $\gord$-downward closures are computable for $\cC$.
\item\label{eq:sep} For every parameterized WQO $\gord$, separability by $\gord$-PTL is decidable for $\cC$.
\end{enumerate}
\end{thm}
This generalizes the two aforementioned results on downward
closures and PTL separability. In addition, \cref{equivalence} applies
to all the examples of $\gord$-PTL described above.

Recall that for each regular language $R$, there is a labeling
automaton $\cA$ such that $R$ is $\autord{\cA}$-upward closed and thus
a $\autord{\cA}$-PTL.  Thus, for languages $K$ and $L$, the following
are equivalent: (i) There \emph{exists} a labeling automaton $\cA$
such that $K$ and $L$ are separable by a $\autord{\cA}$-PTL and (ii)
$K$ and $L$ are separable by a regular language. Already for
one-counter languages, separability by regular languages is
undecidable~\cite{CzerwinskiLasota2017} (for context-free languages,
this was shown in~\cite{SzymanskiWilliams1976,Hunt1982}). However,
\cref{equivalence} tells us that for each fixed $\cA$, separability by
$\autord{\cA}$-PTL is decidable. We make a few applications explicit.

\begin{cor}
Let $\cC$ be a full trio with decidable SUP. For each $d\in\N$,
separability by $\BSS[\MOD[d]]$ is decidable for $\cC$.
\end{cor}

A direct consequence from \cref{equivalence} is that we can decide
whether a regular language is a $\gord$-PTL.  Note that since a
language $L\subseteq\Sigma^*$ is separable from its complement
$\Sigma^*\setminus L$ by some $\gord$-PTL if and only if $L$ is an
$\gord$-PTL itself, \cref{equivalence} implies the following.
\begin{cor}
  Let $\gord$ be a parameterized WQO. Given a regular language $L$, it
  is decidable whether $L$ is an $\gord$-PTL.
\end{cor}

It was shown by Place~et~al.~\cite{PlaceRooijenZeitoun2013b} that for
context-free languages, separability by $\LTT[k]$ is decidable for
each $k\in\N$. Their algorithm uses semilinearity of context-free
languages and Presburger arithmetic.  Here, we extend this result to
all full trios with a decidable SUP.
\begin{cor}\label{ltt}
  Let $\cC$ be a full trio with decidable SUP. For each $k\in\N$,
  separability by $\LTT[k]$ is decidable for $\cC$.
\end{cor}

\subsection*{Separability beyond PTLs}

Our framework can also be applied to separators that do not arise as
PTLs for a particular WQO. This is because we can sometimes apply the
developed ideal representations to separator classes that are infinite
unions of invidual classes of PTLs. For example, consider the fragment
$\BSS[\MOD]$ of first-order logic on words with modular predicates. In
terms of expressible languages, it is the union over all fragments
$\BSS[\MOD[d]]$ with $d\in\N$.  Using a non-trivial algebraic proof,
it was shown by Chaubard, Pin, and
Straubing~\cite{ChaubardPinStraubing2006} that it is decidable whether
a regular language is definable in $\BSS[\MOD]$. Here, we show the
following generalization using a purely combinatorial proof.
\begin{thm}\label{separability:mod:decidable}
Given two regular languages, it is decidable whether they are separable by $\BSS[\MOD]$.
\end{thm}
 
Of course, this raises the question of whether separability by $\BSS[\MOD]$
reduces to the SUP, as it is the case of separability by
$\BSS[\MOD[d]]$ for fixed $d$. However, this is not the case, as is shown here as well.
\begin{thm}\label{separability:mod:undecidable}
Separability by $\BSS[\MOD]$ is undecidable for second-order pushdown languages.
\end{thm}
Since the second-order pushdown languages constitute full
trio~\cite{Maslov1976,aho1968} and have a decidable
SUP~\cite{Zetzsche2015b}, this means separability by $\BSS[\MOD]$ does
not reduce to the SUP.

\section{Computing closures and deciding separability}\label{closure-separability}
In this \lcnamecref{closure-separability}, we present the algorithms
used in \cref{equivalence}. These algorithms work with WQOs on words
under the assumption that these enjoy certain effectiveness
properties. In \cref{ideals}, we will then show that all parameterized WQO
indeed satisfy these properties.
Our algorithms for computing downward closures and deciding
separability rely heavily on the concept of ideals, which have
recently attracted
attention~\cite{LerouxSchmitz2015,FinkelGoubault2009,GoubaultLarrecqSchmitz2016}.
Observe that, in the case of the separability problem, it is always
easy to devise a semi-algorithm for the separability case: We just
enumerate separators--verifying them is possible because we have
decidable emptiness and intersection with regular sets. The difficult
part is to show that inseparability can be witnessed.

These witnesses are always ideals. Let $(X,\gord)$ be a WQO. An
\emph{$\gord$-ascending chain} is a sequence $x_1,x_2,\ldots$ with $x_i\gord
x_{i+1}$ for every $i\in\N$.  
A subset $Y\subseteq X$ is called \emph{($\gord$-)directed} if for any
$x,y\in Y$, there is a $z\in Y$ with $x\gord z$ and $y\gord z$. An
\emph{($\gord$-)ideal} is a non-empty subset $I\subseteq X$ that is
$\gord$-downward closed and $\gord$-directed. Equivalently,
a non-empty subset $I\subseteq X$ is an $\gord$-ideal if $I$ is $\gord$-downward closed and for
any two $\gord$-downward closed sets $Y,Z\subseteq X$ with $I\subseteq
Y\cup Z$, we have $I\subseteq Y$ or $I\subseteq Z$. It is well-known
that every downward closed set can be written as a finite union of
ideals. For more information on ideals,
see~\cite{LerouxSchmitz2015,GoubaultLarrecqSchmitz2016}.

As observed in \cite{GoubaultLarrecqSchmitz2016}, an ideal can witness
inseparability of two languages by belonging to both of their
adherences.  For a set $L\subseteq X$, its \emph{adherence}
$\Adh[\gord]{L}$ is defined as the set of those ideals $I$ of $X$ such
that there exists a directed set $D\subseteq L$ with
$I=\Dclosure[\gord]{D}$.  Equivalently, $I\in\Adh[\gord]{L}$ if and
only if $I\subseteq \Dclosure[\gord]{(L\cap
  I)}$~\cite{LerouxSchmitz2015,GoubaultLarrecqSchmitz2016}.  In this
work, we also use a slightly modified version of adherences in order
to describe ideals of conjunctions of WQOs.  Let $(\gord_s)_{s\in S}$
be a family of well-quasi-orderings on a common set $X$. Moreover, let
$\gord$ denote the conjunction of $(\gord_s)_{s\in S}$. For
$L\subseteq X$, $\Adh[S]{L}$ is the set of all families $(I_s)_{s\in
  S}$ of ideals for which there exists a $\gord$-directed set
$D\subseteq L$ such that $I_s=\Dclosure[\gord_s]{D}$ for each $s\in
S$.

\subsection*{Unboundedness reductions} 
We use counter automata (that are not necessarily count\emph{ing} automata)
to specify unboundedness properties. Let $\cA$ be a counter automaton with
counter set $C$.  Let $\N_\omega=\N\cup\{\omega\}$ and extend $\le$ to
$\N_\omega$ by setting $n<\omega$ for all $n\in\N$.  We define a function
$\bar{\cA}\colon\Sigma^*\to\N_\omega$ by
\[ \bar{\cA}(w)=\sup \left\{\left. \inf_{c\in C} \mu(c) \right| \text{$\cA$ has a run on $w$ arriving at $\mu\in\N^C$}\right\} \]
We say that a counter automaton $\cA$ is \emph{unbounded on
$L\subseteq\Sigma^*$} if for every $k\in\N$, there is a $w\in L$ with
$\bar{\cA}(w)\ge k$.  In other words, iff for every $\nu\in\N^C$, there is a
$w\in L$ such that $\cA$ has a run on $w$ arriving at some  $\mu\ge\nu$.

The following can be shown using a straightforward reduction to the diagonal
problem~\cite{CzerwinskiMartensRooijenZeitounZetzsche2017a,CzerwinskiMartensRooijenZeitoun2015},
which in turn is known to reduce to the SUP~\cite{Zetzsche2015b}.
\begin{lem}\label{unboundedness:decidable}
Let $\cC$ be a full trio with decidable SUP. Then, given a counter automaton
$\cA$ and a language $L$ from $\cC$, it is decidable whether $\cA$ is unbounded
on $L$.
\end{lem}

We are now ready to state the effectiveness assumptions on which our
algorithms rely.  Let $\Sigma$ be an alphabet and $(\Sigma^*,\gord)$
be a WQO. We say that $(\Sigma^*,\gord)$ is an \emph{effective WQO
  with an unboundedness reduction (EWUR)} if the following are
satisfied:
\begin{enumerate}
\item[(a)] For each $w\in\Sigma^*$, the set $\Uclosure[\gord]{w}$ is effectively
regular.
\item[(b)] The set of ideals of $(\Sigma^*,\gord)$ is a recursively enumerable set
of regular languages.
\item[(c)] Given an ideal $I\subseteq \Sigma^*$, one can effectively construct a
counter automaton $\cA_I$ such that for every $L\subseteq\Sigma^*$, $\cA_I$ is
unbounded on $L$ if and only if $I$ belongs to $\Adh[\gord]{L}$.
\end{enumerate}
It should be noted that in order to decide separability by $\gord$-PTL
and compute downward closures, it would have sufficed to require
decidability of adherence membership in full trios with decidable
SUP. The reason why we require the stronger condition (c) is that in
order to show that all parameterized WQOs satisfy these conditions,
we want the latter to be passed on to conjunctions and to WQOs $\gord_f$.

The conditions imply that every upward closed language (hence every downward
closed language) is regular: If $U$ is upward closed, then we can write
$U=\Uclosure[\gord]{\{w_1,\ldots,w_n\}}=\bigcup_{i=1}^n\Uclosure[\gord]{\{w_i\}}$,
which is regular because each $\Uclosure[\gord]{\{w_i\}}$ is regular.
Moreover, we may conclude that given a regular language $R\subseteq\Sigma^*$ it
is decidable whether $R$ is an ideal: If $R$ is an ideal, we find it in an
enumeration; if it is not an ideal, we find words that violate directedness or
downward closedness.

According to the definition of EWUR, we can construct a counter automaton $\cA$
such that $I\in\Adh{L}$ if and only if $\cA$ is unbounded on $L$. Hence,
\cref{unboundedness:decidable} implies the following.
\begin{prop}\label{ewur:adh:decidable}
Let $(\Sigma^*,\gord)$ be an EWUR and let $\cC$ be a full trio with decidable
SUP. Then, given an ideal $I\subseteq\Sigma^*$ and $L\in\cC$, it is decidable
whether $I\in \Adh[\gord]{L}$.
\end{prop}

In \cref{ideals}, we develop ideal representations for all
parameterized WQOs and thus show that they are EWUR.

Let us now sketch how to show \cref{equivalence} assuming that every
parameterized WQO is an EWUR.  The implication \DesImp{eq:dc}{eq:sup}
holds because computing downward closures clearly allows deciding the
SUP. This was shown in~\cite{Zetzsche2015b}.  The implication
\DesImp{eq:sep}{eq:sup} follows
from~\cite{CzerwinskiMartensRooijenZeitounZetzsche2017a}, which
presents a reduction of the SUP to separability by PTL. Thus, it
remains to prove that downward closures are computable and
PTL-separability is decidable for EWUR.  We begin with the former. The
following was shown in \cite{LerouxSchmitz2015}.
\begin{lem}\label{inclusion-vs-adh}
Let $(X,\gord)$ be a WQO and $I_1,\ldots,I_n$ be ideals such that $L\subseteq
I_1\cup \cdots\cup I_n$ and $I_i\not\subseteq I_j$ for $i\ne j$. Then
$I_i\subseteq \Dclosure{L}$ if and only if $I_i\in\Adh{L}$.
\end{lem}

We can now use an algorithm for downward closure computation
from~\cite{GoubaultLarrecqSchmitz2016}, which reduces the computation
to  adherence membership.
\begin{prop}\label{dc:computable}
Let $\cC$ be a full trio with decidable SUP and let $\gord$ be an EWUR.
Then $\gord$-downward closures of languages in $\cC$ are computable.
\end{prop}

We continue with the decidability of separability by $\gord$-PTL for
EWUR $\gord$. We employ the following characterization of separability
in terms of adherences~\cite{GoubaultLarrecqSchmitz2016} for reducing
the separability problem to adherence membership.
\begin{prop}\label{separability}
Let $(X,\gord)$ be a WQO. Then, $K\subseteq X$ and $L\subseteq X$ are separable
by a $\gord$-PTL iff $\Adh[\gord]{K}\cap\Adh[\gord]{L}=\emptyset$.
\end{prop}

We can now use the algorithm from~\cite{GoubaultLarrecqSchmitz2016}
for deciding separability of languages $K$ and $L$ in our setting. By
\cref{separability}, we can use two semi-decision procedures. On the
one hand, we enumerate potential separators $S$ and check whether
$K\subseteq S$ and $L\cap S=\emptyset$. On the other hand, we
enumerate $\gord$-ideals $I$ and check if $I$ belongs to
$\Adh[\gord]{K}\cap\Adh[\gord]{L}$.
\begin{prop}\label{separability:decidable}
Let $\cC$ be a full trio with decidable SUP and $\gord$ be an EWUR.
Then separability by $\gord$-PTL is decidable for $\cC$.
\end{prop}

\section{Ideal representations}\label{ideals}
In this section, we show that every parameterized WQO is an EWUR.
The fact that the subword ordering is an EWUR follows using arguments
from \cite{CzerwinskiMartensRooijenZeitounZetzsche2017a,Zetzsche2015b}.

\begin{prop}\label{ideals:ewur:subword}
The subword ordering $(\Sigma^*,\subword)$ is an EWUR.
\end{prop}

The next step is to show that if $(\Gamma^*,\gord)$ is an EWUR and
$f\colon\Sigma^*\to\Gamma^*$ is a functional transduction, then
$(\Sigma^*,\gord_f)$ is an EWUR. We begin with some general
observations about ideals of WQOs of the shape $\gord_f$, where
$f\colon X\to Y$ is an arbitrary function and $(Y,\gord)$ is a
WQO. First, we describe ideals of $(X,\gord_f)$ in terms of ideals of
$(Y,\gord)$. 

It is easy to see that every ideal of $(X,\gord_f)$ is of the form
form $f^{-1}(J)$ for some ideal $J$ of $(Y,\gord)$. However, a set
$f^{-1}(J)$ is not always an ideal of $(X,\gord_f)$. For example,
suppose $f\colon\Sigma^*\to\N\times\N$ has $f(w)=(|w|,0)$ if $|w|$ is
even and $f(w)=(0,|w|)$ if $|w|$ is odd.  Then $f^{-1}(\N\times\N)$
is not upward directed although $\N\times\N$ is an ideal.
\begin{lem}\label{functions:ideals}
  $I\subseteq X$ is an ideal of $(X,\ford{f})$ if and only if
  $I=f^{-1}(J)$ for some ideal $J$ of $(Y,\gord)$ such that
  $\Dclosure{f(f^{-1}(J))}=J$.
\end{lem}
Note that \cref{functions:ideals} tells us how to represent ideals of
$(X,\gord_f)$ when we have a way of representing ideals of
$(Y,\gord)$. Hence, if the set of ideals of $(\Gamma^*,\gord)$ is
recursively enumerable, then so is the set of ideals of
$(\Sigma^*,\gord_f)$. We will also need to transfer membership in
adherences from $(Y,\gord)$ to $(X,\gord_f)$.
\begin{lem}\label{functions:adherences}
  If $J\subseteq Y$ is an ideal of $(Y,\gord)$ with
  $\Dclosure{f(f^{-1}(J))}=J$, then $f^{-1}(J)\in\Adh{L}$ if and only
  if $J\in\Adh{f(L)}$.
\end{lem}

Equipped with \cref{functions:ideals,functions:adherences}, it is now
straightforward to show that $(\Sigma^*,\gord_f)$ is an EWUR.
\begin{prop}\label{ewur:transducers}
If $(\Gamma^*,\gord)$ is an EWUR and $f\colon \Sigma^*\to\Gamma^*$ is
a functional transducer, then $(\Sigma^*,\gord_f)$ is an EWUR.
\end{prop}

It remains to be shown that being an EWUR is preserved by taking a
conjunction.  Our first step is to characterize which sets are ideals
of a conjunction. Once the statement is found, the proof is relatively
straightforward.
\begin{prop}\label{conjunction-ideals}
Let $S=(\gord_s)_{s\in S}$ be a finite family of WQOs over $X$ and let $(X,\gord)$ be
the conjunction of $S$.  Then $I\subseteq X$ is an ideal of $(X,\gord)$ iff it
can be written as $I=\bigcap_{s\in S} I_s$, where each $I_s$ is an ideal of
$(X,\gord_s)$ and $(I_s)_{s\in S}$ belongs to $\Adh[S]{I}$. 
\end{prop}

The next step describes how to reduce the adherence membership problem for
conjunctions to the adherence membership problem for the participating
orderings. Again, proving the statement is straightforward.
\begin{prop}\label{conjunction-adherence}
Let $S=(\gord_s)_{s\in S}$ be a finite family of WQOs over $X$ and let $(X,\gord)$ be
the conjunction of $S$. Suppose $I_s$ is an $\gord_s$-ideal for each $s\in S$ and $I=\bigcap_{s\in S} I_s$ and that $(I_s)_{s\in S}$ belongs to $\Adh[S]{I}$. Then
$I$ belongs to $\Adh[\gord]{L}$ iff $(I_s)_{s\in S}$ belongs to $\Adh[S]{L}$.
\end{prop}

As expected, a product construction allows us to characterize the
adherence membership for conjunction. 
\begin{lem}\label{ewur:product}
Suppose $(\Sigma^*,\gord_i)$ is an EWUR for $i=1,2$.  Given ideals $I_1$ and
$I_2$ for $\gord_1$ and $\gord_2$, respectively, we can construct a counter
automaton $\cA$ such that for every language $L\subseteq\Sigma^*$, $\cA$ is
unbounded on $L$ iff $(I_1,I_2)$ belongs to $\Adh[\gord_1,\gord_2]{L}$. 
\end{lem}

The following is now a consequence of the previous steps.
\begin{prop}\label{ewur:conjunction}
If $\gord_1$ and $\gord_2$ are EWUR, then their conjunction is an EWUR as well.
\end{prop}

\subsection*{Orderings defined by labeling automata}\label{ideals:auto}
The preceding results already show that every parameterized WQO is an EWUR.
However, since we will study separability by $\BSS[\MOD]$, it will be crucial to
have an explicit, i.e. syntactic representation of ideals of a particular type
of parameterized WQOs, namely those defined by labeling automata. Here, we develop
such a syntax.

Let $\cA$ be a labeling automaton over $\Sigma^*$,
$u_0,\ldots,u_n\in\Sigma^*$, and $v_1,\ldots,v_n\in\Sigma^*$. The word
$w=u_0v_1u_1\cdots v_nu_n$ (more precisely: this particular
decomposition) is a \emph{loop pattern (for $\cA$)} if the run of
$\cA$ on $w$ loops at each $v_i$, $i\in[1,n]$.  In other words, $\cA$
is in the same state before and after reading $v_i$.

\begin{thm}\label{ideals:labeling}
Let $\cA$ be a labeling automaton. The $\autord{\cA}$-ideals are precisely
the sets of the form $\Dclosure[\autord{\cA}]{u_0v_1^*u_1\cdots v_n^*u_n}$,
where $u_0v_1^*u_1\cdots v_n^*u_n$ is a loop pattern for $\cA$.
\end{thm}
By standards arguments about ideals, it is enough to show that those
sets are ideals and that every downward closed set is a finite union
of such sets.

\section{Separability by $\BSS[\MOD]$}\label{sepmod}
In this section, we prove \cref{separability:mod:decidable} and
\cref{separability:mod:undecidable}.  The latter will be shown in
\cref{sec-undecidability} and the former is an immediate consequence
of the following.
\begin{prop}\label{separability:mod:bound}
  Let $\cA_1,\cA_2$ be finite automata with $\le m$ states. 
  $\langof{\cA_1}$ and $\langof{\cA_2}$ are separable by $\BSS[\MOD]$
  if and only if they are separable by $\BSS[\MOD[d]]$, where
  $d=2m^3!$.
\end{prop}
Recall that $\BSS[\MOD[d]]$ are the $\autord{\cM_d}$-PTL, where
$\cM_d$ is the labeling automaton defined on \cpageref{def-md}. From
now on, we write $\modord{d}$ for $\autord{\cM_d}$. \cref{separability:mod:bound} follows from:
\begin{prop}\label{pump-adherence}
Let $\cA_i$ be a finite automaton for $i=1,2$ with $\le m$ states and
let $d$ be a multiple of $2m^3!$.
If 
\[ \Adh[\modord{d}]{\langof{\cA_1}}\cap \Adh[\modord{d}]{\langof{\cA_2}}\ne\emptyset, \]
then 
\[ \Adh[\modord{\ell\cdot d}]{\langof{\cA_1}}\cap \Adh[\modord{\ell\cdot d}]{\langof{\cA_2}}\ne\emptyset \]
for every $\ell\ge 1$.
\end{prop}
The ``if'' direction of \cref{separability:mod:bound} is
trivial and the ``only if'' follows from \cref{pump-adherence}: If
$\langof{\cA_1}$ and $\langof{\cA_2}$ are separable by
$\BSS[\MOD[\ell]]$ for some $\ell\in\N$, then this separator is also
expressible in $\BSS[\MOD[\ell\cdot d]]$. Moreover, together with
\cref{separability}, \cref{pump-adherence} tells us that separability
by $\BSS[\MOD[\ell\cdot d]]$ implies separability by $\BSS[\MOD[d]]$.

The rest of this section outlines the proof of \cref{separability:mod:bound}. Note that according to
\cref{ideals:labeling}, the ideals for $\modord{d}$ are the sets of
the form $I=\Dclosure[\modord{d}]{u_0v_1^*u_1\cdots v_n^*u_n}$ where
$v_i\in(\Sigma^*)^d$. The ideal $I$ belongs to $\Adh[\modord{d}]{L}$
if for each $k\in\N$, there is a word $w\in L$ such that
$u_0v_1^ku_1\cdots v_n^ku_n\modord{d} w$ and $w\in I$. We call such words
$w$ \emph{witness words}.

It is tempting to think that \cref{pump-adherence} just requires a
simple pumping argument: Suppose the ideal $\Dclosure[\modord{d}]
u_1v_1^*u_1\cdots v_n^*u_n$ belongs to the adherence of some
language. Then, we pump the gaps in between embedded letters from the
witness word $u_0v_1^{\ell\cdot k}u_1\cdots v_n^{\ell\cdot k}u_n$.
These gaps, after all, always have length divisible by $d$. For a $d$
with sufficiently many divisors, we would be able to pump the gaps up
to a length divisible by $\ell\cdot d$ so that we can embed
$u_0(v_1^\ell)^ku_1\cdots (v_n^\ell)^ku_n$ via $\modord{\ell\cdot
  d}$. However, in order to show that the $\modord{\ell\cdot d}$-ideal
$I'=\Dclosure[\modord{\ell\cdot d}] u_0(v_1^\ell)^*u_1\cdots
(v_n^\ell)^*u_n$ is contained in the $\modord{\ell\cdot d}$-adherence,
we also have to make sure that resulting witness words are
\emph{members of $I'$}. This makes the proof challenging.

\subsection*{Part I: Small periods}
Our proof of \cref{pump-adherence} consists of three parts.
In the first part, we show that if two regular languages share
an ideal in their adherences, then there exists one in which all
loops (the words $v_i$) are in a certain sense, highly periodic.
Let $\Powerset{\Sigma}$ denote the power set of $\Sigma$ and let
$\Powerset{\Sigma}^{[1,d]}$ denote the set of mappings $\mu\colon
[1,d]\to\Powerset{\Sigma}$. For each word $w\in\Sigma^*$ and $d\in\N$,
let $\kappa_d(w)\in\Powerset{\Sigma}^{[1,d]}$ be defined as follows.
For $i\in[1,d]$, we set
\[ \kappa_d(w)(i)=\{a\in\Sigma \mid \text{$a$ occurs in $w$ at a position $p$}~\text{with $p\equiv i\bmod{d}$}\}. \]
For each word $w\in\Sigma^*$, let $\rho(w)$ be obtained from rotating $w$ by one position to the right.
Hence, for $v\in\Sigma^*$ and $a\in\Sigma$ we have $\rho(va)=av$, and $\rho(\varepsilon)=\varepsilon$. Let $\lambda$ be the inverse map of $\rho$, i.e. rotation to the left. 
For $v\in \Sigma^*$ and $d\in\N$, let $\pi_d(v)\in[1,d]$ be the
smallest $t\in [1,d]$ that divides $d$ such that
$\kappa_d(v)(i+t)=\kappa_d(v)(i)$ for all $i\in[1,d-t]$. Thus, $t$ can
be thought of as a period of $\kappa_d(v)$.
An automaton $\cA=(Q,\Sigma,E,I,F)$ is \emph{cyclic} if $I=F$ and
$|I|=1$.  The first step towards ideals with high periodicity is to
achieve high periodicity in single-loop ideals in cyclic automata:
\begin{lem}\label{reg:smallperiod:cyclic}
Let $\cA_i$ be a cyclic automaton with $\le m$ states for each
$i=1,2$ and let $d$ be a multiple of $m^2!$.  If $\Dclosure[\modord{d}]{v^*}$
belongs to $\Adh[\modord{d}]{\langof{\cA_i}}$ for $i=1,2$, then there is a
$w\in(\Sigma^d)^*$ such that 
(i)~$\Dclosure[\modord{d}]{v^*}\subseteq\Dclosure[\modord{d}]{w^*}$,
(ii)~$\Dclosure[\modord{d}]{w^*}$ also belongs to $\Adh[\modord{d}]{\langof{\cA_i}}$ for $i=1,2$, and
(iii)~$\pi_d(w)\le m^2$.
\end{lem}
The idea is to find in witness words a factor $f$ such that left and
right of $f$, we can pump factors of suitable length. By pumping both
of these factors up by multiplicities that sum up to a constant, we
can essentially move $f$ back and forth and obtain a computation in
which the occurrences of letters in $f$ are spread over all residue
classes modulo some small number $\le m^2$.  

\subparagraph{Associated patterns} In order to extend this to general
ideals and automata, we need more guarantees on how words
$u_0v_1^ku_1\cdots v_n^ku_n$ embed into witness words.

Let $u_0v_1u_1\cdots v_nu_n$ be a loop pattern for $\cM_d$ and let
$L\subseteq\Sigma^*$.  We say that the loop pattern is \emph{associated to $L$}
if for every $k\ge 0$, there is a word $\bar{u}_0\bar{v}_1\bar{u}_1\cdots
\bar{v}_n\bar{u}_n\in L$ such that $v_i^k\modord{d}
\bar{v}_i\in\Dclosure[\modord{d}]{v_i^*}$ for every $i\in[1,n]$ and
$u_i\modord{d}\bar{u}_i\in\Dclosure[\modord{d}]{v_{i}^*u_iv_{i+1}^*}$ for
$i\in[1,n-1]$ and $u_0\modord{d}\bar{u}_0\in\Dclosure[\modord{d}]{u_0v_1^*}$
and $u_n\modord{d}\bar{u}_n\in\Dclosure[\modord{d}]{v_n^*u_n}$.

Of course, if the pattern $u_0v_1u_1\cdots v_nu_n$ is associated to $L$, then
the ideal $I=\Dclosure[\modord{d}]{u_0v_1^*u_1\cdots v_n^*u_n}$ belongs to
$\Adh[\modord{d}]{L}$.  However, the converse is not true. Consider, for
example, the case $d=2$ and the loop pattern $\eword\cdot(aa)\cdot\eword\cdot(abba)\cdot\eword$, where $aa$ and
$abba$ are cycles and the constant parts are all empty.  The resulting ideal
$\Dclosure[\modord{2}]{(aa)^*(abba)^*}$ belongs to
$\Adh[\modord{2}]{(abba)^*}$, just because
$\Dclosure[\modord{2}]{(aa)^*(abba)^*}=\Dclosure[\modord{2}]{(abba)^*}$: Both
sets contain precisely the words in $\{a,b\}^*$ of even length.  Note that the
pattern $\eword\cdot(aa)\cdot\eword\cdot(abba)\cdot\eword$ is not associated to $(abba)^*$, because no word in the
latter contains $(aa)^2$ as an infix, let alone arbitrarily high powers of
$aa$.

However, we will see that every ideal admits a representation by a loop
pattern so that membership in the adherence implies association of the
loop pattern.  A loop pattern $u_0v_1u_1\cdots v_nu_n$ for $\cM_d$ is 
\emph{irreducible} if removing any loop would induce a strictly
smaller ideal. This means, for each $i\in[1,n]$, the loop pattern
$u_0(v_1)u_1\cdots (v_{i-1})u_{i-1}u_i\cdots (v_n)u_n$
induces a strictly smaller ideal than $u_0v_1u_1\cdots v_nu_n$.
Note that every ideal is induced by some irreducible loop pattern:
Just pick one with a minimal number of loops.

\begin{lem}\label{reg:association}
Let $u_0v_1u_1\cdots v_nu_n$ be an irreducible loop pattern for $\cM_d$.  Then
$\Dclosure[\modord{d}]{u_0v_1^*u_1\cdots v_n^*u_n}$ belongs to
$\Adh[\modord{d}]{L}$ if and only if $u_0v_1u_1\cdots v_nu_n$ is associated to
$L$.
\end{lem}
\cref{reg:association} is obtained by first proving that if the loop
pattern is irreducible, then for each $k\in\N$, any embedding of
$u_0v_1^{x_1}u_1\cdots v_n^{x_n}u_n$ into $u_0v_1^{y_1}u_1\cdots
v_n^{y_n}u_n$ for sufficiently large $x_i$ forces at least $k$ copies
of each $v_i$ to be embedded into $v_i^{y_i}$.

Using \cref{reg:association}, we can complete the first proof part:
\begin{lem}\label{reg:smallperiod}
Let $\cA_i$ be a finite automaton with $\le m$ states for each $i=1,2$ and let
$d$ be a multiple of $m^2!$.  If
$\Adh[\modord{d}]{\langof{\cA_1}}\cap\Adh[\modord{d}]{\langof{\cA_2}}\ne\emptyset$,
then there is a loop pattern $u_0v_1u_1\cdots v_nu_n$ for $\cM_d$ such that
$\Dclosure[\modord{d}]{u_0v_1^*u_1\cdots v_n^*u_n}$ belongs to
$\Adh[\modord{d}]{\langof{\cA_i}}$ for $i=1,2$ and $\pi_d(v_i)\le m^2$.
\end{lem}

\subsection*{Part II: Restricting witness words}
In the second part, we place further restrictions on the
structure of ideals that witness inseparability. In return, we get
stronger guarantees on the shape of witness words.  Using \cref{reg:smallperiod}, proving
\cref{pump-adherence} would not be difficult if we could guarantee witness
words of the shape $u_0\bar{v}_1u_1\cdots \bar{v}_nu_n$ with
$\bar{v}_i\in\Dclosure[\modord{d}] v_i^*$ for a pattern
$u_0v_1u_1\cdots v_nu_n$. This is not the case for irreducible loop
patterns: Consider the ideal $I=\Dclosure[\modord{2}] a(abba)^*$.  The
loop pattern $a(abba)$ (with the loop $abba$) is clearly
irreducible. Also, $I$ is a member of $\Adh[\modord{2}]{b\{a,b\}^*}$: For
$k\in\N$, the word $b(abba)^{k+1}\in L$ satisfies $a(abba)^k\modord{2}
b(abba)^{k+1}\modord{2} a(abba)^{k+2}$, which proves
$I\subseteq\Dclosure[\modord{2}]{(L\cap I)}$.  Here, the witness
words $b(abba)^{k+1}$ do not have the above shape.  However, with an extended
syntax for patterns and an adapted irreducibility  notion, we
can guarantee almost that shape.

An \emph{extended loop pattern (for $\cM_d$)} is an expression of the form
$u_0v_1^{[r_1]}u_1\cdots v_n^{[r_n]}u_n$ 
such that
$u_0v_1u_1\cdots v_nu_n$ is a loop pattern for $\cM_d$ (i.e. $v_i\in
(\Sigma^d)^*$ for $i\in[1,n]$) and $r_1,\ldots,r_n\in[0,d-1]$.  The
ideal generated by the pattern is
$\Dclosure[\modord{d}]{u_0v_1^*w_1u_1\cdots v_n^*w_nu_n}$, where $w_i$
is the length-$r_i$ prefix of $v_i$ for $i\in[1,n]$.  Slightly abusing
notation, we use
$\Dclosure[\modord{d}] u_0v_1^{[r_1]}u_1\cdots v_n^{[r_n]}u_n$ to denote the generated ideal.
When we use such an expression 
 with $r_i>d$,  this stands for
$u_1v_1^{[s_1]}u_1\cdots v_n^{[s_n]}u_n$, where $s_i\in[0,d-1]$ and $s_i\equiv r_i\pmod{d}$.

Consider an extended loop pattern $u_0v_1^{[r_1]}u_1\cdots
v_n^{[r_n]}u_n$ for $\cM_d$ and let $w_i$ be the length-$r_i$ prefix
of $v_i$ for $i\in[1,n]$. The pattern is said to be \emph{associated}
to a language $L$ if for every $k\in\N$, there is a word
$\bar{u}_0\bar{v}_1\bar{u}_1\cdots \bar{v}_n\bar{u}_n\in L$ so that
for every $i\in[1,n]$, we have $v_i^kw_i\modord{d} \bar{v}_i$ and
$\bar{v}_i\in\Dclosure[\modord{d}] v_i^{[r_i]}$. Moreover,
$\bar{u}_0=u_0$, $\bar{u}_n=u_n$, and for each $i\in[1,n-1]$: (i)~if
$u_i$ is not empty, then $\bar{u}_i=u_i$ and (ii)~if $u_i$ is empty,
then $\bar{u}_i\in\Dclosure[\modord{d}]
\lambda^{r_i}(v_i)^*v_{i+1}^*$. As in \cref{reg:association}, we have a
notion of irreducible loop patterns, and we show that each ideal is represented by such a
pattern and then obtain:
\begin{lem}\label{reg:association-ext}
The ideal generated by an irreducible extended loop pattern $p$ for $\cM_d$
belongs to $\Adh[\modord{d}]{L}$ if and only if $p$ is associated to $L$.
\end{lem}

We can indeed not guarantee $\bar{u}_i=u_i$ if
$u_i=\varepsilon$ but have to allow for the case
$\bar{u}_i\in\Dclosure[\modord{d}] \lambda^{r_i}(v_i)^*v_{i+1}^*$: The
extended loop pattern $(ab)^{[0]}(cd)^{[0]}$ is irreducible and its
ideal $I=\Dclosure[\modord{2}] (ab)^*(cd)^*$ belongs to
$\Adh[\modord{2}]{(ab)^*ad(cd)^*}$, but in the witness words
$(ab)^kad(cd)^k\in I$, we always have a factor
$ad\in\Dclosure[\modord{2}]{(ab)^*(cd)^*}$.

\subsection*{Part III: Pumping up}
The final part of the proof of \cref{pump-adherence} is to construct
$\modord{\ell\cdot d}$-ideals using pumping.  Here, the strong
guarantees of associated extended loop patterns allow us to focus on
two types of factors in which we must pump: factors $\bar{v}_i$ and
factors $\bar{u}_i$ for empty $u_i$. One can show that repeating
subfactors thereof whose length is divisible by a particular
$\pi_d(v_i)$ will not lead out of the $\modord{\ell\cdot
  d}$-ideal. Moreover, since we established in the first part that
each period $\pi_d(v_i)$ is small ($\le m^2$), we can always find a
factor $f$ of length divisible by $\pi_d(v_i)$ that is pumpable.

\subsection{Undecidability}\label{sec-undecidability}
In this section, we prove \cref{separability:mod:undecidable}.
Second-order pushdown languages are those accepted by second-order
pushdown automata~\cite{Maslov1976} or, equivalently, indexed
grammars~\cite{aho1968}.

In order to prove that separability of
second-order pushdown languages by the fragment $\BSS[\MOD]$ is undecidable, we do not
need a detailed definition of second-order pushdown automata. All we
need is that their languages form a full trio~\cite{aho1968} and that
we can construct automata for two particular types of languages.
Let us describe these languages. For a word $w\in\{1,2\}^*$, let 
$\nu(w)$ be
the number obtained by interpreting the word as a reverse $2$-adic representation. Thus,
for $w\in\{1,2\}^*$, let
$\nu(\varepsilon)=0$, 
$\nu(1w)=2\cdot \nu(w)+1$, 
and   
$\nu(2w)=2\cdot \nu(w)+2$.  
Note that $\nu\colon\{1,2\}^*\to\N$ is a bijection.  
In the full
version of \cite{Zetzsche2015b}, it was shown\footnote{To be precise, this was shown for the unreversed $2$-adic representation, but the reversed case follows by just reversing the images of the morphisms.} that given two
morphisms $\alpha,\beta\colon\Sigma^*\to\{1,2\}^*$, one can
construct in polynomial time an indexed grammar generating
$\{a^{\nu(\alpha(w))}b^{\nu(\beta(w))} \mid w\in\Sigma^+ \}$. Applying a simple transduction
yields the language
\[ L_{\alpha,\beta}=\{a^{\nu(\alpha(w))}cb^{\nu(\beta(w))} \mid w\in\Sigma^+ \} \]
and hence an indexed grammar for $L_{\alpha,\beta}$. Furthermore, the
context-free language $E=\{a^ncb^n \mid n\in\N\}$ is also a second-order
pushdown language.
We apply a technique introduced by Hunt~\cite{Hunt1982}
and simplified by Czerwi\'{n}ski and Lasota~\cite{CzerwinskiLasota2017}.
The idea is to show that \emph{every} decidable problem can be reduced
in polynomial time to our problem:
\begin{prop}\label{undecidability:reduction}
  For each decidable $D\subseteq\Gamma^*$, there is a polynomial-time
  algorithm that, given $u\in\Gamma^*$, computes morphisms
  $\alpha,\beta$ such that $L_{\alpha,\beta}$ is inseparable from $E$ by
  $\BSS[\MOD]$ if and only if $u\in D$.
\end{prop}
Thus, decidability of separability by $\BSS[\MOD]$ would violate the
time hierarchy theorem (see, e.g. \cite[Thm
9.10]{Sipser2012introduction}).  In the proof of
\cref{undecidability:reduction}, we apply the classical reduction from
the halting problem to the PCP. Applied to a terminating TM, this
yields morphisms $\alpha,\beta$, with a bound on the maximal common
prefix of $\alpha(w)$ and $\beta(w)$ for $w\in\Sigma^*$. This implies that in case the input
machine does not accept, $L_{\alpha,\beta}$ and $E$ are separable by $\BSS[\MOD]$.

\subsection*{Future work} The author is confident that the procedure
for separability by $\BSS[\MOD]$ easily extends to separability by
other (albeit less natural) fragments of first-order logic (FO) with
numerical predicates. For example, one could add unary predicates
$\iota$ and $\tau$, where $\iota(x)$ ($\tau(x)$) expresses that $x$ is
the first (last) position. This connects to results of Place and
Zeitoun~\cite{PlaceZeitoun2015}, who developed methods for
transferring decidable separability by a fragment of FO to the
fragment enriched by the successor relation $+1$.  If these methods
could be applied here, this would imply decidable separability by
$\BSS[\MOD,\iota,\tau, +1]$, which is expressively equivalent to the
logic $\BSS[\REG]$. Here, $\REG$ denotes regular predicates of
arbitrary
arity~\cite{ChaubardPinStraubing2006,MacielPeladeauTherien2000}.
 
\subsection*{Acknowledgements} The author is very grateful to
Wojciech Czerwi{\'{n}}ski, Sylvain Schmitz, and Marc Zeitoun
for discussions that yielded important insights.

\bibliographystyle{plain}
\bibliography{bibliography}

\appendix

\section{Proof of \cref{boolean-conjunctions}}
Suppose $S$ consists of the WQOs $\gord_i$ for $i\in[1,n]$.
Every
$\gord$-PTL is an $S$-PTL, because the set $\Uclosure[\gord]{\{w\}}$ can be
written as $\bigcap_{i\in[1,n]}\Uclosure[\gord_i]{\{w\}}$. On the other hand,
every $S$-PTL is a Boolean combination of sets of the form
$\Uclosure[\gord_i]{w}$ with $w\in\Sigma^*$.  Clearly, $\Uclosure[\gord_i]{w}$
is upward closed also with respect to $\gord$ and can thus be written as
$\Uclosure[\gord]{\{w_1,\ldots,w_m\}}$ for some $w_1,\ldots,w_m\in\Sigma^*$,
which is a $\gord$-PTL.

\section{Proof of \cref{unboundedness:decidable}}
\begin{proof}
Let $\cA=(Q,\Sigma,C,E,q_0,F)$. We regard $C$ as an alphabet. Consider the
transducer $T=(Q,\Sigma,C,E',q_0,F)$, where $E'$ is obtained by adding, for
each edge $(q,x,\mu,q')\in E$, an edge $(q,x,u,q')$, where $u\in C^*$ is a word
with $|u|_c=\mu(c)$ for each $c\in C$. Then by definition, $\cA$ is unbounded
on $L$ if and only if for each $n\in\N$, there is a $w\in TL$ with $|w|_c\ge n$
for each $c\in C$.  The latter is an instance of the \emph{diagonal
problem}~\cite{CzerwinskiMartensRooijenZeitoun2015,CzerwinskiMartensRooijenZeitounZetzsche2017a},
which, given a language $K\subseteq\Sigma^*$, asks whether for every $n\in\N$,
there is a $w\in K$ with $|w|_a\ge n$ for all $a\in\Sigma$.  As mentioned
in~\cite{Zetzsche2015b}, for full trios, decidability of the SUP implies
decidability of the diagonal problem, because the former implies computability
of downward closures (with respect to the subword ordering).
\end{proof}

\section{Proof of \cref{inclusion-vs-adh}}
\begin{proof}
Clearly, $I_i\in\Adh{L}$ implies $I_i\subseteq\Dclosure{L}$. Conversely,
suppose $I_1\subseteq\Dclosure{L}$ and $I_1\notin\Adh{L}$. Then there is an
$x\in I_1$ with $x\notin\Dclosure{(L\cap I_1)}$, which means $x\in I_2\cup
\cdots \cup I_n$. We claim that then $I_1\subseteq I_2\cup\cdots\cup I_n$. Let $y\in
I_1$.  There is a $z\in I_1$ with $x\gord z$ and $y\gord z$. Since $x\gord z$,
we have $z\notin \Dclosure{(L\cap I_1)}$ and hence $z\in L_2\cup \cdots\cup
L_n$, which implies $y\in L_2\cup \cdots\cup L_n$. This means $I_1\subseteq
I_2\cup\cdots\cup I_n$ and since $I_1,\ldots,I_n$ are ideals, we have
$I_1\subseteq I_j$ for some $j\in[2,n]$, contrary to our assumption.
\end{proof}

\section{Proof of \cref{dc:computable}}
\begin{proof}
  Given $L$ in $\cC$, we enumerate $\gord$-downward closed
  languages. Since every downward closed set is a finite union of
  ideals, we enumerate finite unions $I_1\cup\cdots\cup I_n$ of
  $\gord$-ideals $I_1,\ldots,I_n$, which is possible because the set
  of ideals is a recursively enumerable set of regular
  languages. Clearly, we only need to enumerate unions where for any
  $i,j\in [1,n]$ with $i\ne j$, we have $I_i\not\subseteq I_j$.

  It remains to check whether $\Dclosure[\gord]{L}=I_1\cup\cdots\cup
  I_n$. Note that $\Dclosure[\gord]{L}\subseteq I_1\cup\cdots\cup I_n$
  if and only if $L\subseteq I_1\cup\cdots \cup I_n$, so that we can
  check whether $L\cap (\Sigma^*\setminus (I_1\cup\cdots\cup
  I_n))=\emptyset$. The latter is decidable because the decidability
  of the SUP implies the decidability of the emptiness problem and
  $\cC$ is effectively closed under intersection with regular
  languages.

  The other inclusion is more interesting. Suppose we have already
  established $\Dclosure[\gord]{L}\subseteq I_1\cup\cdots\cup I_n$.
  Then, according to \cref{inclusion-vs-adh}, we have
  $I_i\subseteq\Dclosure[\gord]{L}$ if and only if
  $I_i\in\Adh[\gord]{L}$.  We can therefore apply
  \cref{ewur:adh:decidable} to check whether the latter holds.
\end{proof}

\section{Proof of \cref{separability:decidable}}
\begin{proof}
  Suppose we are given languages $K$ and $L$.  We decide separability
  by combining two semi-algorithms. One enumerates $\gord$-PTL and for
  each such language $R$, decides whether $K\subseteq R$ and $L\cap
  R=\emptyset$.  If such an $R$ is found, the languages are reported
  separable. The other semi-algorithm enumerates ideals $I$ of
  $(\Sigma^*,\gord)$ and then, via \cref{ewur:adh:decidable}, decides
  whether $I\in\Adh[\gord]{K}$ and $I\in\Adh[\gord]{L}$. If such an
  ideal $I$ is found, the languages are reported inseparable. The
  correctness and termination of this algorithm is guaranteed by
  \cref{separability}.
\end{proof}

\section{Proof of \cref{ideals:ewur:subword}}
\begin{proof}
  Of course, for every $w\in\Sigma^*$, $\Uclosure[\subword]{w}$ is
  effectively regular.  Moreover, it is well-known that the ideals of
  $(\Sigma^*,\subword)$ are exactly the languages of the form
  $\{a_0,\eword\}\Gamma_1^*\{a_1,\eword\}\cdots
  \Gamma_n^*\{a_n,\eword\}$, where $a_0,\ldots,a_n\in\Sigma$ and
  $\Gamma_1,\ldots,\Gamma_n\subseteq\Sigma$~\cite{Jullien1969}.
  Lastly, if $I=\{a_0,\eword\}\Gamma_1^*\{a_1,\eword\}\cdots
  \Gamma_n^*\{a_n,\eword\}$, we build $\cA_I$ as follows.  For each
  $i\in[1,n]$, choose a word $w_i\in\Gamma_i^*$ that contains each
  letter of $\Gamma_i$ exactly once. Then, it is easy to construct
  $\cA_I$ so that $\bar{\cA_I}(w)\ge k$ if and only if $w\in I$ and
  $a_0w_1^ka_1\cdots w_n^ka_n\subword w$. Then clearly $\cA_I$ is
  unbounded on $L$ if and only if we have $I\subseteq\Dclosure[\subword]{(L\cap
    I)}$. The latter is equivalent to $I\in\Adh[\subword]{L}$.
\end{proof}

\section{Proof of \cref{functions:ideals}}
\begin{proof}
  If $I\subseteq X$ is an ideal, then the set $J:=\Dclosure{f(I)}$ is
  downward closed by definition and upward directed because $I$ is.
  Hence, $J$ is an ideal. Moreover, $I=f^{-1}(J)$, because $I\subseteq
  f^{-1}(J)$ is immediate and $f^{-1}(J)\subseteq I$ holds because $I$
  is downward closed. This also implies $\Dclosure{f(f^{-1}(J))}=\Dclosure{f(I)}=J$.

  Conversely, suppose $I=f^{-1}(J)$ for an ideal $J\subseteq Y$ with
  $\Dclosure{f(f^{-1}(J))}=J$.  First, $I=f^{-1}(J)$ is downward
  closed because $J$ is. Moreover, we have $\Dclosure{f(I)}=J$, which
  means given $x,y\in I$, we can find a common upper bound $z\in J$ for
  $f(x)\in J$ and $f(y)\in J$ and then a $z'\in f(I)$ with
  $z\gord z'$. Then $z'=f(w)$ for some $w\in I$ and hence $x\gord_f w$
  and $y\gord_f w$. Thus $I$ is upward directed.
\end{proof}

\section{Proof of \cref{functions:adherences}}
\begin{proof}
  Suppose $f^{-1}(J)\in\Adh{L}$, equivalently, $f^{-1}(J)\subseteq
  \Dclosure{(L\cap f^{-1}(J))}$. We show that $J\subseteq
  \Dclosure{(f(L)\cap J)}$.  For $y\in J$, we can find $y'\in
  f(f^{-1}(J))$ with $y\gord y'$. Say $y'=f(x')$ with $x'\in
  f^{-1}(J)$.  Thus, there is $x''\in L\cap f^{-1}(J)$ with $x'\gord_f
  x''$. Since $y\gord y'=f(x')\gord f(x'')\in f(L)\cap J$, we have
  shown $J\subseteq\Dclosure{(f(L)\cap J)}$.

  Conversely, suppose $J\in\Adh{f(L)}$, hence
  $J\subseteq\Dclosure{(f(L)\cap J)}$.  This means, for $x\in
  f^{-1}(J)$, we can find $x'\in L$ with $f(x)\gord f(x')$ and
  $f(x')\in J$.  Thus, $f^{-1}(J)\subseteq\Dclosure{(L\cap f^{-1}(J))}$
  and hence $f^{-1}(J)\in\Adh{L}$.
\end{proof}

\section{Proof of \cref{ewur:transducers}}
\begin{proof}
  First, for every $w\in\Sigma^*$, we have
  $\Uclosure[\gord_f]{w}=f^{-1}(\Uclosure[\gord]{f(w)})$, which is
  effectively regular because $\Uclosure[\gord]{f(w)}$ is.

  Second, \cref{functions:ideals} tells us that the ideals of
  $(\Sigma^*,\gord_f)$ are precisely the sets of the form $f^{-1}(I)$
  where $I\subseteq\Gamma^*$ is an ideal of $(\Gamma^*,\gord)$ and for
  which $\Dclosure[\gord]{f(f^{-1}(I))}=I$. Therefore, the set of
  ideals of $(\Sigma^*,\gord_f)$ is recursively enumerable: Enumerate
  the ideals $I$ of $(\Gamma^*,\gord)$ and check whether
  $\Dclosure[\gord]{f(f^{-1}(I))}=I$. The latter is possible because
  $f(f^{-1}(I))\subseteq\Gamma^*$ is effectively regular (regular
  languages are closed under rational transductions) and because for
  the EWUR $(\Gamma^*,\gord)$, we can effectively compute a finite
  automaton for the downward closure $\Dclosure[\gord]{f(f^{-1}(I))}$:
  The regular languages constitute a full trio with decidable SUP.
  Thus, we can compare the regular languages
  $\Dclosure[\gord]{f(f^{-1}(I))}$ and $I$.

  Third, given an ideal $J\subseteq\Sigma^*$ (represented as a finite
  automaton), we can find an ideal $I\subseteq\Gamma^*$ with
  $J=f^{-1}(I)$. Since $(\Gamma^*,\gord)$ is an EWUR, we can compute a
  counter automaton $\cA_I$ such that $\cA_I$ is unbounded on a
  language $L\subseteq\Gamma^*$ if and only if $I\in\Adh[\gord]{L}$.
  According to \cref{functions:adherences}, we know that
  $J\in\Adh[\gord_f]{K}$ if and only if $I\in\Adh[\gord]{f(K)}$, which
  in turn is equivalent to $\cA_I$ being unbounded on $f(K)$. We can
  thus construct $\cA_J$ as a product of $\cA_I$ and the transducer
  for $f$ so that $\cA_J(w)=\cA_I(f(w))$ for every
  $w\in\Sigma^*$. Clearly, $\cA_J$ is unbounded on $K$ if and only if
  $\cA_I$ is unbounded on $f(K)$.
\end{proof}

\section{Proof of \cref{conjunction-ideals}}
\begin{proof}
Let $I\subseteq X$ be an ideal of $(X,\gord)$. Then $I$ is directed with respect
to $\gord_s$ for each $s\in S$.  Hence, $I_s=\Dclosure[\gord_s]{I}$ is an ideal for
each $s\in S$.
We claim that $I=\bigcap_{s\in S} I_s$. Clearly,
$I\subseteq\Dclosure[\gord_s]{I}=I_s$, hence $I\subseteq\bigcap_{s\in S}I_s$.  On
the other hand, if $x\in \bigcap_{s\in S} I_s$, then for each $s\in S$, there
is a $x_s\in I$ with $x\gord_s x_s$.  Since $I$ is directed, we find a $y\in I$
with $x_s\gord y$ for each $s\in S$. Hence, in particular $x\gord_s y$. This
implies $x\gord y$ and thus $x\in I$.
This proves $I=\bigcap_{s\in S} I_s$.
Finally, as a $\gord$-directed set, $I$ itself witnesses that $(I_s)_{s\in S}$
belongs to $\Adh[S]{I}$.

Conversely, suppose $I=\bigcap_{s\in S} I_s$ and that $(I_s)_{s\in S}$ belongs
to $\Adh[S]{I}$. The latter means that there is a $\gord$-directed set
$D\subseteq I$ such that for each $s\in S$, we have $I_s=\Dclosure[\gord_s]{D}$.
We claim that $I=\Dclosure[\gord]{D}$. If $x\in I$, then for each $s\in S$, there
is an $x_s\in D$ with $x\gord_s x_{s}$.  Since $S$ is finite and $D$ is
$\gord$-directed, we find a $y\in D$ with $x_s\gord y$ for all $s\in S$. Then for
each $s\in S$, we have $x\gord_s x_{s}\gord_s y$ and thus $x\gord y$. Hence,
$I\subseteq \Dclosure[\gord]{D}$.  On the other hand, if $x\gord y$ for $y\in D$,
then clearly $x\gord_s y$ for each $s\in S$ and thus $x\in \bigcap_{s\in S}
I_s=I$.
\end{proof}

\section{Proof of \cref{conjunction-adherence}}
\begin{proof}
Let $D\subseteq I$ be a $\gord$-directed set with $I_s=\Dclosure[\gord_s]{D}$
for every $s\in S$.  Suppose $I\in\Adh[\gord]{L}$.  Then there is a
$\gord$-directed set $D'\subseteq L$ with $I=\Dclosure[\gord]{D'}$. We claim
that $I_s=\Dclosure[\gord_s]{D'}$.  For $x\in I_s$, there is a $y\in D$ with
$x\gord_s y$. Since $y\in I$, there is a $z\in D'$ with $y\gord z$. In
particular, we have $x\gord_s z\in D'$. This proves ``$\subseteq$''.  On the
other hand, we know $D'\subseteq I\subseteq I_s$, which implies
$\Dclosure[\gord_s]{D'}\subseteq I_s$, since $I_s$ is $\gord_s$-downard closed.

Conversely, suppose that $(I_s)_{s\in S}$ belongs to $\Adh[S]{I}$ with a
directed set $D'\subseteq L\cap I$ such that $I_s=\Dclosure[\gord_s]{D'}$. We
claim that $I=\Dclosure[\gord]{D'}$. Of course, we have the inclusion
``$\supseteq$'' because $D'\subseteq I$, so assume $x\in I$. Since
$I_s=\Dclosure[\gord_s]{D'}$ and $I=\bigcap_{s\in S} I_s$, for each $s\in S$,
there is a $y_s\in D'$ with $x\gord_s y_s$. The $\gord$-directedness of $D'$
yields a $y\in D'$ with $y_s\gord y$ for every $s\in S$. Then in particular
$x\gord y$ and hence $x\in\Dclosure[\gord]{D'}$.
\end{proof}

\section{Proof of \cref{ewur:product}}
\begin{proof}
Let $\cA_i=(Q_i,\Sigma,C_i,E_i,q^i_{0},F_i)$ be a counter automaton that
characterizes adherence membership of $I_i$ with respect to $\gord_i$ for
$i=1,2$.  We construct a product automaton $\cA$ so that $\cA$ has states
$Q_1\times Q_2$, counters $C_1\cup C_2$, and satisfies $(q^1_0,q^2_0,\eword,0) \autsteps[\cA]
(q^1,q^2,w,\mu)$ if and only if $(q_0^i,\eword,0)\autsteps (q^i,w,\mu|_{C_i})$
for $i=1,2$. Moreover, $\cA$ has final states $F_1\times F_2$.  

We claim that $\cA$ is unbounded on $L$ if and only if $(I_1,I_2)$ belongs to
$\Adh[\gord_1,\gord_2]{L}$.  We will use the fact that when a counter automaton
$\cB$ is unbounded on $K\cup L$, then it is unbounded on $K$ or on $L$. Suppose
$\cA$ is unbounded on $L$. By construction, unboundedness of $\cA$ implies
unboundedness of $\cA_1$ and of $\cA_2$. Therefore, $\cA$ must be unbounded on
$L\cap I_1$: Otherwise, $\cA$, and thus $\cA_1$, would be unbounded on $L\setminus I_1$, which
is impossible by definition of $\cA_1$.  By the same argument, $\cA$ must be unbounded on $L\cap I_1\cap
I_2$.  Then, $\cA$ is also unbounded on some sequence $w_1,w_2,\ldots\in L\cap
I_1\cap I_2$ and since $\gord$ is a WQO, we may assume that this sequence is a
$\gord$-chain. Therefore, the $\gord$-directed set $D=\{w_i\mid i\ge 1\}$
satisfies $D\subseteq I_1\cap I_2$ and $I_i\subseteq
\Dclosure[\gord_i]{D}$ for $i=1,2$. This proves $(I_1,I_2)\in\Adh[\gord_1,\gord_2]{L}$.

Conversely, suppose $(I_1,I_2)\in\Adh[\gord_1,\gord_2]{L}$. Then there is a
$\gord$-directed set $D\subseteq L$ with $I_i=\Dclosure[\gord_i]{D}$. This
implies that $\cA_1$ and $\cA_2$ are unbounded on $D$. Hence, there are
sequences $u_1,u_2,\ldots\in D$ and $v_1,v_2,\ldots\in D$ such that
 $\cA_1$ is unbounded on
$u_1,u_2,\ldots$ and $\cA_2$ is unbounded on $v_1,v_2,\ldots$. Thus, we have
$I_1\subseteq\Dclosure[\gord_1]{\{u_i \mid i\ge 1\}}$ and
$I_2\subseteq\Dclosure[\gord_2]{\{v_i \mid i\ge 1\}}$. Since $D$ is $\gord$-directed,
we can successively find elements $w_1,w_2,\ldots\in D$ such that $u_i\gord
w_i$ and $v_i\gord w_i$ and $w_i\gord w_{i+1}$. Then we have 
$I_i\subseteq\Dclosure[\gord_i]{\{w_k\mid k\ge 1\}}$ for $i=1,2$ and 
since $D\subseteq I_1\cap I_2$,
we have $\Dclosure[\gord_i]{\{w_k\mid k\ge 1\}}=I_i$.

Hence, $\cA_1$ and $\cA_2$ are both unbounded on $w_1,w_2,\ldots$.  We can
therefore pick a subsequence $w'_1,w'_2,\ldots$ such that $\bar{\cA_1}(w'_k)\ge
k$ for $k\ge 1$. As an infinite subsequence of $w_1,w_2,\ldots$, this sequence
will still satisfy $\Dclosure[\gord_2]{\{w'_k\mid k\ge 1\}}=I_2$ and in
particular, $\cA_2$ is unbounded on $w'_1,w'_2,\ldots$. We can therefore find
another subsequence $w''_1,w''_2,\ldots$ such that $\bar{\cA_i}(w''_k)\ge k$
for every $k\ge 1$ and $i\in\{1,2\}$. Thus, $\cA$ is unbounded on
$w''_1,w''_2,\ldots$ and hence on $L$.
\end{proof}

\section{Proof of \cref{ewur:conjunction}}
\begin{proof}
Let $\gord$ be the conjunction of $\gord_1$ and $\gord_2$.  First, for
$w\in\Sigma^*$, we have $\Uclosure[\gord]{w}=\Uclosure[\gord_1]{w}\cap
\Uclosure[\gord_2]{w}$, so that $\Uclosure[\gord]{w}$ inherits effective
regularity from $\Uclosure[\gord_1]{w}$ and $\Uclosure[\gord_1]{w}$.

According to \cref{conjunction-ideals}, we can represent an ideal $I$ of
$\gord$ by a pair $(I_1,I_2)$ such that $I_i$ is an ideal for $\gord_i$,
$I=I_1\cap I_2$, and $(I_1,I_2)\in\Adh[\gord_1,\gord_2]{I}$. Hence, in order to
show that the set of ideals of $\gord$ is a recursively enumerable set of
regular languages, we need to show that it is decidable whether
$(I_1,I_2)\in\Adh[\gord_1,\gord_2]{I}$. To this end, we use \cref{ewur:product}
to construct a counter automaton $\cA$ that is unbounded on $L$ if and only if
$(I_1,I_2)\in\Adh[\gord_1,\gord_2]{L}$. Since $I=I_1\cap I_2$ is effectively
regular, we can decide whether $\cA$ is unbounded on $I$ using
\cref{unboundedness:decidable}.
\end{proof}

\section{Proof of \cref{ideals:labeling}}
Note that every unambiguous automaton $\cA$ defines an order $\autord{\cA}$ on
$\langof{\cA}$ in the same way labeling automata define an order on $\Sigma^*$.
We will now also use $\autord{\cA}$ to denote this order. We say that an
unambiguous automaton $\cB$ is a \emph{subautomaton} of $\cA$ if $\cB$ is
obtained from $\cA$ by deleting some edges.  The following can be shown,
roughly speaking, by decomposing $\cB$ into strongly connected components and
dividing $\langof{\cB}$ according to which path through the resulting graph a
word takes.
\begin{lem}\label{autorder:funion:single}
For a subautomaton $\cB$ of an unambiguous automaton $\cA$, 
$\langof{\cB}$ is a finite union of sets of the form
\[ \Dclosure[\autord{\cA}]{u_0v_1^*u_1\cdots v_n^*u_n}, \]
 where $u_0v_1u_1\cdots v_nu_n$ is a loop pattern for $\cA$.
\end{lem}
\begin{proof}
We decompose $\cB$ into its directed acyclic graph $G$ of strongly connected
components and notice that this graph has only finitely many paths.  Moreover,
for each strongly connected component $C$ and and states $p$ and $q$ in $C$,
there are only finitely many simple paths from $p$ to $q$.  Every run through
$C$ from $p$ to $q$ can be reduced to one of these simple paths by deleting
loops.  Therefore, we can divide the set $\langof{\cB}$ according to which paths
in $G$ they a word follows and to which simple paths in each component it reduces.
This yields a decomposition of 
$\langof{\cB}$ as a finite union of sets of the form $u_0L_1u_1\cdots L_nu_n$
such that there are states $q_0,\ldots,q_n$ so that 
\begin{itemize}
\item $q_0$ is initial and $q_n$ is final,
\item for $i\in[0,n]$, either $(q_i,u_i,q_{i+1})$ is an edge in $\cB$, or $u_i=\eword$ and $q_{i+1}=q_i$,
\item for $i\in[1,n]$, $L_i$ is the set of words read on a cycle from $q_i$ to $q_i$.
\end{itemize}
For each $i\in [1,n]$, consider the strongly connected component of $\cB$ that
contains $q_i$ and let $E_i$ be the set of edges of $\cB$ in this component.

There exists a word $v_i\in L_i$ whose run from $q_i$ to $q_i$ (note that there
is at most one such run because $\cA$ is a labeling automaton) uses every edge
from $E_i$ at least once: For each $e\in E_i$, take a run from $q_i$ to $q_i$
that uses $e$.  Then take $v_i$ to be the word read on the concatenation of all
these runs.

We claim that $u_0L_1u_1\cdots L_nu_n=\Dclosure[\autord{\cA}] u_0v_1^*u_1\cdots
v_n^*u_n$.  Since $u_0L_1u_1\cdots L_nu_n$ is clearly downward closed with
respect to $\autord{\cA}$ and contains $u_0v_1^*u_1\cdots v_n^*u_n$, the
inclusion ``$\supseteq$'' holds.  Conversely, suppose $w_i\in L_i$ for $i\in
[1,n]$. Consider a particular $i\in [1,n]$ and let $r=e_1\cdots e_k\in E^*$ be
the run of $\cB$ when reading $w_i$ from $q_i$ to $q_i$.  Each $e_j$ occurs in
the run $s\in E^*$ of $v_i$, so that the run $s^k$ of $v_i^k$ contains
$e_1\cdots e_k$ as a subsequence and we can write $s^k=t_0e_1t_1\cdots e_kt_k$
for some $t_0,\ldots,t_k\in E^*$.  Since $e_i$ ends in the state where
$e_{i+1}$ starts and $r$ and $s^k$ are both cycles from $q_i$ to $q_i$, every
run $t_i$ is a cycle. This implies that $u_0w_1u_1\cdots w_nu_n\autord{\cA}
u_0v_1^{|w_1|}u_1\cdots v_n^{|w_n|}u_n$.  This proves the inclusion
``$\subseteq$''.
\end{proof}

We shall prove that the ideals of $(\Sigma^*,\autord{\cA})$ are precisely those sets of the form
$\Dclosure[\autord{\cA}]{u_0v_1^*u_1\cdots v_n^*u_n}$. The first step in
proving that is to show that every downward closed language is a finite union
of such sets.  Here, we will use the fact that ideals of the subword ordering
are precisely the languages $\{a_0,\eword\}\Gamma_1^*\{a_1,\eword\}\cdots
\Gamma_n^*\{a_n,\eword\}$, where $a_0,\ldots,a_n\in\Sigma$ and
$\Gamma_1,\ldots,\Gamma_n\subseteq\Sigma$~\cite{Jullien1969}.
\begin{prop}\label{autord:funion}
Let $\cA$ be a labeling automaton and $L\subseteq\Sigma^*$.  The set
$\Dclosure[\autord{\cA}]{L}$ is a finite union of sets of the form
\[ \Dclosure[\autord{\cA}]{u_0v_1^*u_1\cdots v_n^*u_n}, \]
 where $u_0v_1u_1\cdots
v_nu_n$ is a loop pattern for $\cA$.
\end{prop}
\begin{proof}
Let $\cA=(Q,\Sigma,E,I,F)$.
For each $p,q\in Q$, we define $K_{p,q}=\{w\in L\mid \sigma_{\cA}(w)=(p,q)\}$. Then
we have 
\[ \Dclosure[\autord{\cA}]{L}=\bigcup_{p,q\in Q} \Dclosure[\autord{\cA}]{K_{p,q}}. \]
 Therefore, it suffices to consider the case
that there are fixed $p,q\in Q$ such that for every $u,v\in L$, we have $\sigma_{\cA}(u)=(p,q)$.  Note that then $u\autord{\cA} v$ if and only if
$\cA(u)\subword\cA(v)$ for $u,v\in L$.  Let $\seqof[p,q]{\cA}$ denote the set
of all runs of $\cA$ that start in $p$ and end in $q$.  Let $\pi\colon E^*\to\Sigma^*$ be the projection onto labels of edges. Observe that
$\Dclosure[\autord{\cA}]{L}=\pi((\Dclosure{\cA(L)})\cap \seqof[p,q]{\cA})$.
(Here, $\Dclosure{\cA(L)}$ denotes the downward closure with respect to the subword ordering.)

The language $\Dclosure{\cA(L)}$ is a finite union of sets of the form
$e_0E_1^*e_1\cdots E_n^*e_n$, where $E_i\subseteq E$ and $e_i\in
E\cup\{\eword\}$.  Hence, we would like to prove the \lcnamecref{autord:funion}
for sets of the form $\pi(e_0E_1^*e_1\cdots E_n^*e_n\cap \seqof[p,q]{\cA})$.
However, these are not necessarily downward closed. Therefore, we prove that
\[ \Dclosure[\autord{\cA}]{\pi(e_0E_1^*e_1\cdots E_n^*e_n\cap \seqof[p,q]{\cA})} \]
can be written as a finite union of sets
$\Dclosure[\autord{\cA}]{u_0v_1^*u_1\cdots v_n^*u_n}$.

The set $e_0E_1^*e_1\cdots E_n^*e_n\cap \seqof[p,q]{\cA}$ is a finite union of sets
of the form $e_0S_1e_1\cdots S_n e_n$ such that there are states $q_0,\ldots,q_{n+1}$
so that
\begin{itemize}
\item for $i\in[0,n]$, either $e_i=\eword$ and $q_{i+1}=q_i$, or $e_i$ is
an edge from $q_i$ to $q_{i+1}$,
\item for $i\in[1,n]$, $S_i\subseteq E_i^*$ is the set of runs of $\cA$
from $q_i$ to $q_{i+1}$ that only use edges in $E_i$.
\end{itemize}
Therefore, it suffices to show that $\Dclosure[\autord{\cA}]{\pi(e_0S_1e_1\cdots S_n e_n)}$ can be written as a finite union as desired.
\newcommand{\funionprod}[3]{u_{#1,0}v_{#1,1}^{#3}u_{#1,1}\cdots v_{#1,#2}^{#3}u_{#1,#2}}
Let $\cA_i$ be the unambiguous automaton obtained from $\cA$ by making
$q_i$ the only initial state and $q_{i+1}$ the only final state. Moreover, let
$\cB_i$ be obtained from $\cA_i$ be removing all edges outside of $E_i$.  Then,
we have have $\pi(S_i)=\langof{\cB_i}$. According to
\cref{autorder:funion:single}, $S_i=\langof{\cB_i}$ is a finite union of sets of
the form $\Dclosure[\autord{\cA_i}]{u_0v_1^*u_1\cdots v_k^*u_k}$, where
$u_0v_1u_1\cdots v_ku_k$ is a loop pattern for $\cA_i$. Therefore, 
our set $\Dclosure[\autord{\cA}]{\pi(e_0S_1e_1\cdots S_n e_n)}$ is a finite union of sets of the form
\begin{align} \Dclosure[\autord{\cA}]{\left(\pi(e_0)(\Dclosure[\cA_1]{I_1})\pi(e_1)\cdots (\Dclosure[\autord{\cA_n}]{I_n})\pi(e_n)\right)}, \label{summand} \end{align}
where $I_i=\funionprod{i}{k_i}{*}$ for $i\in[1,n]$.
The definition of $\autord{\cA}$ implies immediately that \cref{summand} equals
\[ \Dclosure[\autord{\cA}]{   \big(\pi(e_0)\big(\funionprod{1}{k_1}{*}\big)\pi(e_1)%
\cdots \big(\funionprod{n}{k_n}{*}\big) \pi(e_n)   \big)}. \]
Moreover, 
\[ \pi(e_0)\funionprod{1}{k_1}{}\pi(e_1) %
\cdots \funionprod{n}{k_n}{} \pi(e_n) \]
 is clearly a loop pattern for $\cA$ (where the $v_{i,j}$ play the
role of the $v_i$).
\end{proof}

We are now ready to prove \cref{ideals:labeling}.
\begin{proof}[Proof of \cref{ideals:labeling}]
Let us show that the language \[ I=\Dclosure[\autord{\cA}]{u_0v_1^*u_1\cdots v_n^*u_n} \]
is in fact an $\autord{\cA}$-ideal.  It is clearly $\autord{\cA}$-downward
closed.  Consider the word $w_k=u_0v_1^ku_1\cdots v_n^ku_n$ for each $k\in\N$.
Then we have $w_0\autord{\cA}w_1\autord{\cA}\cdots$, so that the set
$D=\{w_k\mid k\in\N\}$ is $\autord{\cA}$-directed. Moreover,
$I=\Dclosure[\autord{\cA}]{D}$, which proves that $I$ is the
$\autord{\cA}$-downward closure of a $\autord{\cA}$-directed set and hence an
$\autord{\cA}$-ideal.

It remains to be shown that every ideal is of the above form. Let $I$ be an
ideal of $\autord{\cA}$.  In \cref{autord:funion} we have seen that every
downward closed is a finite union of sets of the above form.  In particular, we
can write $I=I_1\cup\cdots\cup I_k$, where each $I_k$ is of the above form.
However, since $I$ is an ideal and the $I_i$ are downward closed, this implies
that for some $i\in[1,n]$, we have $I\subseteq I_i$ and thus $I=I_i$.
\end{proof}

\section{Proofs for \cref{sepmod}}

\begin{lem}\label{reg:ideal-charact}
Suppose $v\in(\Sigma^d)^*$. Then $\Dclosure[\modord{d}]{v^*}=\{w\in
(\Sigma^d)^* \mid \kappa_d(w)\subseteq \kappa_d(v)\}$.
\end{lem}
\begin{proof}
Let $u\in\Dclosure[\modord{d}]{v^*}$, say $w\modord{d} v^k$. Then clearly
$w\in(\Sigma^d)^*$. Moreover, if $a\in\Sigma$ occurs at a position $p$ in $w$
with $p\equiv i\pmod{d}$, then $a$ occurs at some position $p+d\N$ in $v$.
Hence, $\kappa_d(w)\subseteq\kappa_d(v)$.

Suppose $w\in(\Sigma^d)^*$ and $\kappa_d(w)\subseteq\kappa_d(v)$. Write
$w=a_1\cdots a_n$, $a_1,\ldots,a_n\in\Sigma$.  Since
$a_i\in\kappa_d(w)(i)\subseteq \kappa_d(v)(i)$, each $a_i$ occurs at some
position $p$ in $v$ with $p\equiv i\pmod{d}$. Hence, we can write $v=x_ia_iy_i$
with $|x_i|\equiv i-1 \pmod{d}$ and therefore $|y_i|\equiv |v|-|x_i|-1\equiv
d-i \pmod{d}$. In particular, $|y_ix_{i+1}|\equiv (d-i)+i=d$.  Moreover,
$|x_1|\equiv 0\mod{d}$ and $y_n\equiv d-n\equiv 0\pmod{d}$. Therefore,
\[ w=a_1\cdots a_n\modord{d} \overline{x_1}a_1\overline{y_1x_2}a_2\overline{y_2x_3}\cdots \overline{y_{n-1}x_n}a_n\overline{y_n}=v^n \]
where $\overline{u}$ expresses that $u\in(\Sigma^d)^*$.  Thus
$w\in\Dclosure[\modord{d}]{v^*}$.
\end{proof}

\begin{lem}\label{reg:ideal-inclusion}
Suppose $v,w\in(\Sigma^d)^*$.
Then $\Dclosure[\modord{d}]{v^*}\subseteq\Dclosure[\modord{d}]{w^*}$ if and only if $\kappa_d(v)\subseteq\kappa_d(w)$.
\end{lem}
\begin{proof}
If $\Dclosure[\modord{d}]{v^*}\subseteq \Dclosure[\modord{d}]{w^*}$, then in particular
$v\in \Dclosure[\modord{d}]{w^*}$ and thus $\kappa_d(v)\subseteq \kappa_d(w)$ by
\cref{reg:ideal-charact}. 

Suppose $\kappa_d(v)\subseteq\kappa_d(w)$. Since $v\in(\Sigma^d)^*$, we have
$\kappa_d(v^n)=\kappa_d(v)$ for any $n\in\N$ and hence $v^n\in
\Dclosure[\modord{d}]{w^*}$ by \cref{reg:ideal-charact}. This implies
$\Dclosure[\modord{d}]{v^*}\subseteq \Dclosure[\modord{d}]{w^*}$.
\end{proof}

\begin{lem}\label{reg:insert-periodic}
If $\kappa_d(xyz)\subseteq\kappa_d(v)$ and $\pi_d(v)$ divides $|y|$, then
$\kappa_d(xyyz)\subseteq\kappa_d(v)$.
\end{lem}
\begin{proof}
Let $i\in[1,d]$. We will show that $\kappa_d(xyyz)(i)\subseteq\kappa_d(v)(i)$. 
Hence, let $a\in\kappa_d(xyyz)(i)$. Then there is a position $p\in[1,|xyyz|]$
with $p\equiv i\pmod{d}$ such that the $p$-th position of $xyyz$ reads $a$.

If $p\in[1,|xy|]$, we are done, so assume $p\in[|xy|+1,|xyyz|]$. Then, $a$ also
occurs at position $q=p-|y|$ in $xyz$. This means, if $j\equiv q\pmod{d}$, then
$a\in\kappa_d(xyz)(j)\subseteq\kappa_d(v)(j)$. Observe that $i\equiv p\pmod{d}$
implies $i\equiv p\pmod{\pi_d(v)}$ and thus $i\equiv p=q+|y|\equiv q\equiv
j\pmod{\pi_d(v)}$. Therefore, we have $a\in\kappa_d(v)(j)=\kappa_d(i)$.
\end{proof}

\begin{lem}\label{reg:pump-inside-ideal-ext}
Suppose $\pi_d(v)$ divides $|y|$ and $|y|$ divides $d$.  If $xyz\in
\Dclosure[\modord{d}]{v^{[r]}}$, then for every $\ell\in\N$, $xy^{1+\ell\cdot
d/|y|}z\in\Dclosure[\modord{d}]{v^{[r]}}$.
\end{lem}
\begin{proof}
Let $w=xy^{1+\ell\cdot d/|y|}z$. Since $d$ divides $|xyz|$, it also divides
$|w|=|xyz|+(\ell\cdot d/|y|)\cdot |y|$. According to \cref{reg:ideal-charact-ext},
we have $\kappa_d(xyz)\subseteq\kappa_d(v)$. An $(\ell\cdot d/|y|)$-fold
application of \cref{reg:insert-periodic} tells us that
$\kappa_d(xy^{1+\ell\cdot d/|y|}z)\subseteq\kappa_d(v)$. Now, \cref{reg:ideal-charact-ext}
states that $xy^{1+\ell\cdot d/|y|}z\in\Dclosure[\modord{d}]{v^*}$.
\end{proof}

\begin{proof}[Proof of \cref{reg:smallperiod:cyclic}]
Write $v=v_1\cdots v_n$, $v_1,\ldots,v_n\in\Sigma$.  Since
$\Dclosure[\modord{d}]{v^*}$ belongs to $\Adh[\modord{d}]{\langof{\cA_i}}$ for
$i=1,2$, we have
$v\in\Dclosure[\modord{d}]{(\langof{\cA_i}\cap\Dclosure[\modord{d}]{v^*})}$ for
$i=1,2$.  This means there are words $v^{(i)}=u_0^{(i)}v_1u_1^{(i)}\cdots
v_nu_n^{(i)}\in\langof{\cA_i}\cap \Dclosure[\modord{d}]{v^*}$  such that
$u_j^{(i)}\in(\Sigma^d)^*$ for $j\in[1,n]$ and $i=1,2$. Note that since
$v^{(i)}\in\Dclosure[\modord{d}]{v^*}$ and $v\modord{d}v^{(i)}$, we have
$\Dclosure[\modord{d}]{(v^{(i)})^*}=\Dclosure[\modord{d}]{v^*}$ and thus
$\kappa_d(v^{(i)})=\kappa_d(v)$ according to \cref{reg:ideal-inclusion}.

In the run of $\cA_i$ for $u_0^{(i)}v_1u_1^{(i)}\cdots v_nu_n^{(i)}$, let
$q_j^{(i)}$ be the state occupied after reading $u_j^{(i)}$, for $j\in[0,n]$
and $i=1,2$.  Since $m^2!$ divides $d$, which in turn divides $n$, we have
$n+1>m^2!\ge m^2$.  Therefore, there are $j,k\in[0,n]$, $j<k$, with
$(q^{(1)}_j,q^{(2)}_j)=(q^{(1)}_k,q^{(2)}_k)$.  Moreover, they can be chosen so
that $t:=k-j<m^2$.  Since $m^2!$ divides $d$, we know that $t<m^2$ divides $d$
and may define $r=d/t$.  Let $x_i=u_0^{(i)}v_1u_1^{(i)}\cdots v_ju_j^{(i)}$,
$y_i=v_{j+1}u_{j+1}^{(i)}\cdots v_ku_k^{(i)}$, $z_i=v_{k+1}u_{k+1}^{(i)}\cdots
v_nu_n^{(i)}$. Then, by the choice of $j,k$, we have $(x_iy_i^*z_i)^*\subseteq\langof{\cA_i}$.
In particular, the word 
\[ w_i=\prod_{\ell=0}^{r-1} x_iy_iy_i^{\ell}z_ix_iy_iy_i^{r-\ell}z_i \] 
belongs to $\langof{\cA_i}$. Moreover, since $|y_i|=t+\sum_{\ell=j+1}^k
|u_\ell^{(i)}|\equiv t\bmod{d}$, we can conclude
\[ |w_i|=r\cdot (2\cdot |x_iy_iz_i|+r\cdot |y_i|)\equiv r\cdot (2\cdot |v^{(i)}|+d)\equiv 0 \bmod{d},\]
which implies $w_i\in(\Sigma^d)^*$.
We claim that 
\[\kappa_d(w_i)=\bigcup_{\ell=0}^{r-1}\kappa_d(\rho^{\ell t}(v^{(i)})). \]

We begin with the inclusion ``$\supseteq$''.  Note that for each $\ell\in
[0,r-1]$ and $i\in\{1,2\}$, 
\begin{itemize}
\item the word $x_i$ occurs in $w_i$ at a position $p$ with $p\equiv |x_iy_iz_i|+\ell t\pmod{d}$
and hence $p\equiv \ell t\pmod{d}$,  
\item the word $y_i$ occurs in $w_i$ at a position $p$ with $p\equiv |x_i|+\ell t\pmod{d}$,  
\item the word $z_i$ occurs in $w_i$ at a position $p$ with $p\equiv |x_iy_i|+\ell t\pmod{d}$.
\end{itemize}
Hence, for each position $p$ in $v^{(i)}$ and each $\ell\in[0,r-1]$, there is a position
$p'\equiv p+\ell t\pmod{d}$ with 
$\kappa_d(v^{(i)})(p)\subseteq\kappa_d(w_i)(p')$. This prove the inclusion ``$\supseteq$''.

On the other hand, every factor $x_i$, $y_i$, and $z_i$ that occurs in the definition
of $w_i$ at a position $p\in [1,|w_i|]$ also occurs in $v^{(i)}$ at a position
$p'\in[1,n]$ with $p'\equiv p-\ell t\pmod{d}$ for some $\ell\in
[0,r-1]$. Therefore, we also have the inclusion ``$\subseteq$''.

The identity $\kappa_d(w_i)=\bigcup_{\ell=0}^{r-1}\kappa_d(\rho^{\ell\cdot
t}(v^{(i)}))$ clearly implies that $\pi_d(w_i)\le t$ and also
$\Dclosure[\modord{d}]{(v^{(i)})^*}\subseteq\Dclosure[\modord{d}]{w_i^*}$,
which in turn yields
$\Dclosure[\modord{d}]{v^*}\subseteq\Dclosure[\modord{d}]{w_i^*}$. Moreover,
since $(x_iy_i^*z_i)^*\subseteq\langof{\cA_i}$, we have
$w_i^*\subseteq\langof{\cA_i}$ and in particular
$\Dclosure[\modord{d}]{w_i^*}\subseteq \Dclosure[\modord{d}]{\langof{\cA_i}}$.
This clearly implies  that $\Dclosure[\modord{d}]{w_i^*}$ belongs to
$\Adh[\modord{d}]{\langof{\cA_i}}$ for $i=1,2$.  Hence, if we can show
$\Dclosure[\modord{d}]{w_1^*}=\Dclosure[\modord{d}]{w_2^*}$, the proof is
complete. We use $\rho$ also as a rotation map on $\Powerset{\Sigma}^{[1,d]}$:
For $\mu\in\Powerset{\Sigma}^{[1,d]}$ and $i\in[1,d]$, let $\rho(\mu)(i)=\mu(i')$,
where $i'\in[1,d]$ is chosen so that $i'\equiv i-1\bmod{d}$.
Observe that since $\kappa_d(v^{(i)})=\kappa_d(v)$ for $i\in\{1,2\}$,
we have
\[ \kappa_d(w_i)=\bigcup_{\ell=0}^{r-1}\kappa_d(\rho^{\ell t}(v^{(i)}))=\bigcup_{\ell=0}^{r-1}\rho^{\ell t}(\kappa_d(v^{(i)}))=\bigcup_{\ell=0}^{r-1}\rho^{\ell t}(\kappa_d(v)), \]
and thus $\kappa_d(w_1)=\kappa_d(w_2)$, which, according to \cref{reg:ideal-inclusion},
implies $\Dclosure[\modord{d}]{w_1^*}=\Dclosure[\modord{d}]{w_2^*}$.
\end{proof}

\subsection{Proof of \cref{reg:association}}
Suppose $x,y\in\Sigma^*$, $x=x_1\cdots x_r$, $y=y_1\cdots y_s$,
$x_1,\ldots,x_r,y_1,\ldots,y_r\in\Sigma$.  A strictly monotone map
$\alpha\colon \{1,\ldots,r\}\to\{1,\ldots,s\}$ is a \emph{$d$-embedding of $x$
in $y$} if $r\equiv s\pmod{d}$, $x_i=y_{\alpha(i)}$ for $i\in[1,r]$, and for each $i\in[1,r]$, 
we have $\alpha(i)\equiv i\pmod{d}$. Clearly, we have $x\modord{d}
y$ if and only if there is a $d$-embedding of $x$ in $y$.  Now let
$u_0v_1u_1\cdots v_nu_n$ be a loop pattern for $\cM_d$ and
$x=u_0v_1^{x_1}u_1\cdots v_n^{x_n}u_n$ and $y=u_0v_1^{y_1}u_1\cdots
v_n^{y_n}u_n$. Then a $d$-embedding of $x$ in $y$ is called \emph{$k$-normal} if
for each $i\in[1,n]$, $\alpha$ maps at least $k$-many factors $v_i$ in $x$ to $v_i^{y_i}$. Clearly,
if $k\le x_i\le y_i$ for all $i\in[1,n]$, then there exists a normal $d$-embedding
of $x$ in $y$. However, not every $d$-embedding has to be $k$-normal.
\begin{lem}\label{reg:normalembeddings}
Let $u_0v_1u_1\cdots v_nu_n$ be an irreducible loop pattern for $\cM_d$. For each
$k\in\N$, there is a constant $\ell\in\N$ such that if $\alpha$ is a
$d$-embedding of $u_0v_1^{x_1}u_1\cdots v_n^{x_n}u_n$ in $u_0v_1^{y_1}u_1\cdots
v_n^{y_n}u_n$ and $x_i\ge \ell$ for $i\in[1,n]$, then $\alpha$ is $k$-normal.
\end{lem}
\begin{proof}
Let us call a $d$-embedding \emph{$(k,i)$-normal} if it maps at least $k$-many factors
$v_i$ in $x$ into the factor $v_i^{y_i}$ in $y$.
To simplify notation, we will
always write $x$ and $y$ for the words $x=u_0v_1^{x_1}u_1\cdots v_n^{x_n}u_n$
and $y=u_0v_1^{y_1}u_1\cdots v_n^{y_n}u_n$.

Suppose the contrary. 
Then there is a $k\in\N$ such that for every $\ell\in\N$, there
are $x_1,\ldots,x_n\in\N$ and $y_1,\ldots,y_n\in\N$ with $x_i\ge\ell$ for $i\in[1,n]$
such that there is a $d$-embedding of $x$ in $y$ that is not $(k,j)$-normal for some $j\in[1,n]$.
Among the $j$ for which this occurs, one has to occur infinitely often. Hence, there is a
$k\in\N$ and a $j\in[1,n]$ such that for every $\ell\in\N$, there
are $x_1,\ldots,x_n\in\N$ and $y_1,\ldots,y_n\in\N$ with $x_i\ge\ell$ for $i\in[1,n]$
such that there is a $d$-embedding of $x$ in $y$ that is not $(k,j)$-normal.

If a $d$-embedding is not $(k,j)$-normal, then all but at most $(k-1)+2$ factors
$v_j$ must be mapped either to the factor $u_0v_1^{y_1}u_1\cdots v_{j-1}^{y_{j-1}}u_{j-1}$
or to the factor $u_{j+1}v_{j+2}^{y_{j+2}}u_{j+2}\cdots v_n^{y_n}u_n$: At most $k-1$
factors are mapped to $v_j^{y_j}$ and at most two further factors are partially
mapped to $v_j^{y_j}$.  Therefore, we have at least one of the following cases:
\begin{enumerate}
\item for each $\ell\in\N$, there are $x_1,\ldots,x_n$ and $y_1,\ldots,y_n$ with $x_i\ge
k$ for $i\in[1,n]$ such that there is a $d$-embedding of $x$ in $y$ that maps
at least $\ell$ factors $v_j$ to $u_0v_1^{y_1}u_1\cdots v_{j-1}^{y_{j-1}}u_{j-1}$.
\item for each $\ell\in\N$, there are $x_1,\ldots,x_n$ and $y_1,\ldots,y_n$ with $x_i\ge
k$ for $i\in[1,n]$ such that there is a $d$-embedding of $x$ in $y$ that maps
at least $\ell$ factors $v_j$ to $u_{j+1}v_{j+2}^{y_{j+2}}u_{j+2}\cdots v_n^{y_n}u_n$.
\end{enumerate}
Let us consider the first case (the second can be treated the same way). 
We claim that this implies 
\begin{align} \Dclosure[\modord{d}]{u_0v_1^*u_1\cdots v_n^*u_n}=\Dclosure[\modord{d}]{u_0v_1^*u_1\cdots v_{j-1}^{*}u_{j-1}u_j\cdots v_n^*u_n}. \label{eq:nonminimal}\end{align}
The inclusion ``$\supseteq$'' clearly holds. For the other direction, consider
$u_0v_1^{z_1}u_1\cdots v_n^{z_n}u_n$. Then there are
$x_1,\ldots,x_n,y_1,\ldots,y_n\in\N$ such that $x_i\ge z_i$ and there exists a
$d$-embedding of $x$ into $y$ that maps at least $z_j$ factors $v_j$ into
$u_0v_1^{y_1}u_1\cdots v_{j-1}^{y_{j-1}}u_{j-1}$. 
This means we have \[ u_0v_1^{z_1}u_1\cdots v_{j-1}^{z_{j-1}}u_{j-1}v_j^{z_j}\modord{d} u_0v_1^{y_1}u_1\cdots v_{j-1}^{y_{j-1}}u_{j-1} \] and hence
\[ u_0v_1^{z_1}u_1\cdots v_n^{z_n}u_n \modord{d} u_0v_1^{y_1}u_1\cdots v_{j-1}^{y_{j-1}}u_{j-1} u_j v_{j+1}^{z_{j+1}}\cdots v_n^{z_n}u_n. \]
since clearly $u_j v_{j+1}^{z_{j+1}}\cdots v_n^{z_n}u_n\modord{d} u_j
v_{j+1}^{z_{j+1}}\cdots v_n^{z_n}u_n$ and $\modord{d}$ is multiplicative.
This implies the inclusion ``$\subseteq$'' of \cref{eq:nonminimal}.
Finally, note that  \cref{eq:nonminimal} contradicts the
assumed irreducibility.
\end{proof}

\begin{proof}[Proof of \cref{reg:association}]
Clearly, if a loop pattern is associated with a language, then its induced
ideal belongs to the adherence of the language.  Conversely, suppose the ideal
$I=\Dclosure[\modord{d}]{u_0v_1^*u_1\cdots v_n^*u_n}$ belongs to
$\Adh[\modord{d}]{L}$. Let $k\in\N$ and $x_1,\ldots,x_n\ge k$ and let
$\ell\in\N$ be the constant provided by \cref{reg:normalembeddings}. Without
loss of generality, we may assume that $\ell\ge k$.

Since $I$ belongs to $\Adh[\modord{d}]{L}$, there is a word $w\in L$ such that
$u_0v_1^\ell u_1\cdots v_n^\ell u_n\modord{d} w\modord{d} u_0v_1^{y_1}u_1\cdots
v_n^{y_n}u_n$ for some $y_1,\ldots,y_n\in\N$. This means in particular that
there is a $d$-embedding $\alpha$ of $u_0v_1^\ell u_1\cdots v_n^\ell u_n$ into
$w$ and a $d$-embedding $\beta$ of $w$ into $u_0v_1^{y_1}u_1\cdots
v_n^{y_n}u_n$. By composing these two $d$-embeddings, we obtain a $d$-embedding
$\gamma$ of $u_0v_1^\ell u_1\cdots v_n^\ell u_n$ into the word $u_0v_1^{y_1}u_1\cdots
v_n^{y_n}u_n$. By the choice of $\ell$, $\gamma$ has to be $k$-normal. This means
that $\gamma$ maps at least $k$ copies of $v_i$ to $v_i^{y_i}$ for each
$i\in[1,n]$.  We can therefore decompose $w=\bar{u}_0\bar{v}_1\bar{u}_1\cdots
\bar{v}_n\bar{u}_n$ so that these $k$ copies of $v_i$ that $\gamma$ maps to
$v_i^{y_i}$ are mapped by $\alpha$ to $\bar{v}_i$ and $|\bar{v}_i|$ is
divisible by $d$.

Since $\beta$ maps $\bar{v}_i$ to $v_i^{y_i}$, we have
$\bar{v}_i\in\Dclosure[\modord{d}]{v_i^*}$. This also implies that $\beta$ maps
$\bar{u}_0$ to $u_0v_1^{y_1}$, and $\beta$ maps $\bar{u}_i$ to
$v_i^{y_i}u_iv_{i+1}^{y_{i+1}}$, and $\beta$ maps $\bar{u}_n$ to
$v_n^{y_n}u_n$. Moreover, $\alpha$ maps $u_i$ to $\bar{u}_i$ for each
$i\in[0,n]$.
In other words, we have $v_i^k\modord{d} \bar{v}_i\in\Dclosure[\modord{d}]{v_i^*}$
for every $i\in[1,n]$ and
$u_i\modord{d}\bar{u}_i\in\Dclosure[\modord{d}]{v_{i}^*u_iv_{i+1}^*}$ for
$i\in[1,n-1]$ and $u_0\modord{d}\bar{u}_0\in\Dclosure[\modord{d}]{u_0v_1^*}$
and $u_n\modord{d}\bar{u}_n\in\Dclosure[\modord{d}]{v_n^*u_n}$. Thus, $I$ is
associated to $L$.
\end{proof}

\subsection{Proof of \cref{reg:smallperiod}}
\begin{proof}
Suppose $I$ belongs to $\Adh[\modord{d}]{\langof{\cA_i}}$ for $i=1,2$.  Let
$u_0v_1u_1\cdots v_nu_n$ be an irreducible loop pattern for $\cM_d$ such that
$I=\Dclosure[\modord{d}]{u_0v_1^*u_1\cdots v_n^*u_n}$. According to \cref{reg:association},
the loop pattern $u_0v_1u_1\cdots v_nu_n$ is associated to $\langof{\cA_i}$ for $i=1,2$.

In particular, there is a word $\bar{u}_{i,0}\bar{v}_{i,1}\bar{u}_{i,1}\cdots
\bar{v}_{i,n}\bar{u}_{i,n}\in\langof{\cA_i}$ such that
$v_j^m\modord{d}\bar{v}_{i,j}\in\Dclosure[\modord{d}]{v_j^*}$ for $j\in[1,n]$ and
$i=1,2$ and
$u_j\modord{d}\bar{u}_{i,j}\in\Dclosure[\modord{d}]{v_{j}^*u_jv_{j+1}^*}$ for
$j\in[1,n-1]$ and $u_0\modord{d}\bar{u}_{i,0}\in\Dclosure[\modord{d}]{u_0v_1^*}$
and $u_n\modord{d}\bar{u}_{i,n}\in\Dclosure[\modord{d}]{v_n^*u_n}$.

We can therefore write $\bar{v}_{i,j}=t_{i,j,1}\cdots t_{i,j,m}$ with
$v_j\modord{d}t_{i,j,\ell}\in\Dclosure[\modord{d}]{v_j^*}$. Consider the run of
$\cA_i$ on the word \[ \bar{u}_{i,0}\bar{v}_{i,1}\bar{u}_{i,1}\cdots
\bar{v}_{i,n}\bar{u}_{i,n}. \] Since $\cA_i$ has $\le m$ states, for each
$j\in[1,n]$, this run must occupy the same before and after reading some infix
$t_{i,j,\ell}\cdots t_{i,j,k}$.  Let $q_{i,j}$ be this state and let
$\bar{v}_{i,j}=x_{i,j}y_{i,j}z_{i,j}$ be the decomposition so that
$y_{i,j}=t_{i,j,\ell}\cdots t_{i,j,k}$. Then we have
$v_j\modord{d}y_{i,j}\in\Dclosure[\modord{d}]{v_j^*}$ and also
$x_{i,j},z_{i,j}\in\Dclosure[\modord{d}]{v_j^*}$. The former implies that
$\Dclosure[\modord{d}]{y_{i,j}^*}=\Dclosure[\modord{d}]{v_j^*}$.

Let $\cA_{i,j}$ be the automaton obtained from $\cA_i$ by making $q_{i,j}$ the
only initial and final state. Then $\cA_{i,j}$ is cyclic and we have
$y_{i,j}^*\subseteq\langof{\cA_{i,j}}$. In particular, the ideal
$\Dclosure[\modord{d}]{v_j^*}=\Dclosure[\modord{d}]{y_{i,j}^*}$ belongs to
$\Adh[\modord{d}]{\langof{\cA_{i,j}}}$. Now \cref{reg:smallperiod:cyclic}
yields a $w_{j}\in(\Sigma^d)^*$ such that 
\begin{itemize}
\item $\Dclosure[\modord{d}]{v_j^*}\subseteq \Dclosure[\modord{d}]{w_j^*}$,
\item $\Dclosure[\modord{d}]{w_j^*}$ belongs to $\Adh[\modord{d}]{\langof{\cA_{i,j}}}$,
\item $\pi_d(w_j)\le m^2$.
\end{itemize}
We claim that $u_0w_1u_1\cdots w_nu_n$ is a loop pattern as desired in the
\lcnamecref{reg:smallperiod}. It remains to show that
$\Dclosure[\modord{d}]{u_0w_1^*u_1\cdots w_n^*u_n}$ belongs to
$\Adh[\modord{d}]{\langof{\cA_i}}$ for $i=1,2$.

Let $k\in\N$. Since $\Dclosure[\modord{d}]{w_j^*}$ belongs to $\Adh[\modord{d}]{\langof{\cA_{i,j}}}$ for $i\in\{1,2\}$ and $j\in[1,n]$, there is a word
$w'_{i,j}\in\langof{\cA_i}$ such that
$w_j^k\modord{d}w'_{i,j}\in\Dclosure[\modord{d}]{w_j^*}$.  Define
\[ t=\bar{u}_{i,0}x_{i,1}w'_{i,1}z_{i,1}\bar{u}_{i,1}\cdots x_{i,n}w'_{i,n}z_{i,n}\bar{u}_{i,n}.  \]
 Then we have
$u_0\bar{w}_1^ku_1\cdots \bar{w}_n^ku_n\modord{d} t\in\langof{\cA_i}$.
Moreover, since $x_{i,j},z_{i,j}\in\Dclosure[\modord{d}]{v_j^*}\subseteq
\Dclosure[\modord{d}]{\bar{w}_j^*}$ and by the choice of the $\bar{u}_{i,j}$,
the word $t$ is contained in
$\Dclosure[\modord{d}]{u_0\bar{w}_1^*u_1\cdots \bar{w}_n^*u_n}$. This proves that
$\Dclosure[\modord{d}]{u_0\bar{w}_1^*u_1\cdots \bar{w}_n^*u_n}$ belongs to
$\Adh[\modord{d}]{\langof{\cA_i}}$ for $i=1,2$ and hence completes the
\lcnamecref{reg:smallperiod}.
\end{proof}

\subsection{Proof of \cref{reg:association-ext}}

\begin{lem}\label{reg:ideal-charact-ext}
Suppose $v\in(\Sigma^d)^*$. Then every $r\in[0,d-1]$:
\[ \Dclosure[\modord{d}] v^{[r]}=\{u\in\Sigma^* \mid |u|\equiv r\bmod{d},~\kappa_d(u)\subseteq\kappa_d(v)\}. \]
\end{lem}
\begin{proof}
Let $w$ be the length-$r$ prefix of $v$.
Let $u\in\Dclosure[\modord{d}]{v^{[r]}}$, say $u\modord{d} v^kw$. Then clearly
$|u|\equiv r\bmod{d}$. Moreover, if $a\in\Sigma$ occurs at a position $p$ in $u$
with $p\equiv i\pmod{d}$, then $a$ occurs at some position $p+d\N$ in $v$.
Hence, $\kappa_d(u)\subseteq\kappa_d(v)$.

Suppose $u\in\Sigma^*$ with $|u|\equiv r\bmod{d}$ and $\kappa_d(u)\subseteq\kappa_d(v)$. Write
$u=a_1\cdots a_n$, $a_1,\ldots,a_n\in\Sigma$.  Since
$a_i\in\kappa_d(u)(i)\subseteq \kappa_d(v)(i)$, each $a_i$ occurs at some
position $p$ in $v$ with $p\equiv i\bmod{d}$. Hence, we can write $v=x_ia_iy_i$
with $|x_i|\equiv i-1 \bmod{d}$ and therefore $|y_i|\equiv |v|-|x_i|-1\equiv
d-i \bmod{d}$. In particular, $|y_ix_{i+1}|\equiv (d-i)+i=d\bmod{d}$.  Moreover,
$|x_1|\equiv 0\bmod{d}$ and $|y_nw|\equiv d-n+r\equiv 0\bmod{d}$. Therefore,
\[ u=a_1\cdots a_n\modord{d} \overline{x_1}a_1\overline{y_1x_2}a_2\overline{y_2x_3}\cdots \overline{y_{n-1}x_n}a_n\overline{y_nw}=v^nw \]
where $\overline{z}$ expresses that $z\in(\Sigma^d)^*$.  Thus
$u\in\Dclosure[\modord{d}]{v^{[r]}}$.
\end{proof}

Consider an extended loop pattern $u_0v_1^{[r_1]}u_1\cdots
v_n^{[r_n]}u_n$ and let $w_i$ be the length-$r$ prefix of $v_i$ for
$i\in[1,n]$. We say that this extended loop pattern is \emph{irreducible} if
\begin{enumerate}
\item the corresponding loop pattern $u_0(v_1)w_1u_1\cdots (v_n)w_nu_n$ is
  irreducible and
\item for each $i\in[0,n-1]$, $u_i$ is either empty or the last letter
  of $u_i$ is not contained in $\kappa_d(v_{i+1})(d)$ and
\item for each $i\in[1,n]$, $u_i$ is either empty or the first letter
  of $u_i$ is not contained in $\kappa_d(v_i)(r_i+1)$.
\end{enumerate}

\begin{lem}\label{reg:make-extended-irreducible}
  Let $x_0y^{[s_1]}_1\cdots y^{[s_\ell]}_\ell x_\ell$ be an extended
  loop pattern for $\cM_d$ for which $\pi_d(y_i)\le m$ for every
  $i\in[1,\ell]$.  Then there is an irreducible extended loop pattern
  $u_0v_1^{[r_1]}u_1\cdots v_n^{[r_n]}u_n$ for $\cM_d$ generating the
  same ideal where also $\pi_d(v_i)\le m$ for every $i\in[1,n]$.
\end{lem}
\begin{proof}
  We define the \emph{length} of an extended loop pattern
  $u_0v_1^{[r_1]}u_1\cdots v_n^{[r_n]}u_n$ to be $|u_0|+\cdots
  |u_n|+n\cdot d$. In other words, each loop $v_i$ contributes $d$ to
  the length. 

  Let $I$ be the ideal $\Dclosure[\modord{d}]x_0y_1^{[s_1]}x_1\cdots
  y_\ell^{[s_\ell]}x_\ell$. Furthermore, let $u_0v_1^{[r_1]}u_1\cdots
  v_n^{[r_n]}u_n$ be an extended loop pattern of minimal length $N$ among
  all extended loop patterns that generate $I$ and for which
  $\pi_d(v_i)\le m$ for every $i\in[1,n]$. Let $w_i$ be the
  length-$r_i$ prefix of $v_i$ for $i\in[1,n]$. 

  By minimality, the loop pattern $u_0(v_1)w_1u_1\cdots (v_n)w_nu_n$ has to be irreducible:
  Otherwise, there would be a loop $v_i$ such that 
  \[ I=\Dclosure[\modord{d}] u_0v_1^*w_1u_1 \cdots v_{i-1}^*w_{i-1}u_{i-1}w_iu_i\cdots v_n^*w_nu_n \]
  and hence the extended loop pattern 
  \[ u_0v_1^{[r_i]}u_1 \cdots v_{i-1}^{[r_{i-1}]}u_{i-1}w_iu_i\cdots v_n^{[r_n]}u_n \]
  would generate $I$ and have length $N-d+r_i<N$.

  Now consider some non-empty $u_i$ and suppose its first letter is
  contained in $\kappa_d(v_i)(r_i+1)$. In other words,
  $u_i=a\bar{u}_i$ with $a\in\kappa_d(v_i)(r_i+1)$. Then we could
  replace $v_i^{[r_i]}u_i$ by $v_i^{[r_i+1]}\bar{u}_i$. The resulting
  extended loop pattern clearly generates the same ideal. Moreover,
  the requirement for periods would still be met. Finally, this
  extended loop pattern would have length $N-1$, in contradiction to
  minimality.
  
  Now consider some non-empty $u_i$ and suppose its last letter is
  contained in $\kappa_d(v_{i+1})(d)$. In other words,
  $u_i=\bar{u}_ia$ with $a\in\kappa_d(v_{i+1})(d)$. Then we could
  replace the term $u_iv_{i+1}^{[r_{i+1}]}$ by
  $\bar{u}_i\lambda(v_{i+1})^{[r_{i+1}+1]}$. It is easy to see that
  this would result in an extended loop pattern that generates the
  same ideal. Moreover, we would have
  $\pi_d(\lambda(v_{i+1}))=\pi_d(v_{i+1})\le m$. Finally, this
  extended loop pattern would have length $N-1$, contradicting
  minimality.
\end{proof}

\begin{proof}[Proof of \cref{reg:association-ext}]
  Clearly, if the ideal generated by $p$ is associated to $L$, then it
  belongs to $\Adh[\modord{d}]{L}$.

  Conversely, let $p=u_0v_1^{[r_1]}u_1\cdots v_n^{[r_n]}u_n$ be an
  extended loop pattern for $\cM_d$ and suppose its generated ideal
  $I$ belongs to $\Adh[\modord{d}]{L}$. Let $w_i$ be the length-$r_i$
  prefix of $v_i$ for $i\in[1,n]$.  Since the loop pattern
  $u_0(v_1)w_1u_1\cdots (v_n)w_nu_n$ (the loop parts are in brackets)
  is irreducible, it is associated to $L$ according to
  \cref{reg:association}. 

  Thus, for given $k\in\N$, we find a word
  \begin{equation} w=\tilde{u}_0\tilde{v}_1\tilde{u}_1\cdots\tilde{v}_n\tilde{u}_n\in L \label{association-ext-decomp}\end{equation}
 such that
  $v_i^{k+1}\modord{d}\tilde{v}_i\in\Dclosure[\modord{d}]v_i^*$ for every
  $i\in[1,n]$ and $w_iu_i\modord{d}\tilde{u}_i\in\Dclosure[\modord{d}]
  v_i^*w_iu_iv_{i+1}^*$ for $i\in[1,n-1]$ and
  $u_0\modord{d}\tilde{u}_0\in\Dclosure[\modord{d}] u_0v_1^*$ and
  $w_nu_n\modord{d}\tilde{u}_n\in\Dclosure[\modord{d}] v_n^*w_nu_n$.

  In the first step, we modify the decomposition
  \cref{association-ext-decomp} of $w$ by moving, for each
  $i\in[1,n]$, the last $d-r_i$ letters of $\tilde{v}_i$ to its right
  neighbor $\tilde{u}_i$. Let the resulting decomposition be
  \[ w=\hat{u}_0\hat{v}_1\hat{u}_1\cdots \hat{v}_n\hat{u}_n. \]

  Since $v_i^{k+1}\modord{d}\tilde{v}_i\in\Dclosure[\modord{d}] v_i^*$ 
  and $w_iu_i\modord{d}\tilde{u}_i\in\Dclosure[\modord{d}] v_i^*w_iu_iv_{i+1}^*$, we now have
  \begin{enumerate}
  \item 
    $v_i^kw_i\modord{d}\hat{v}_i\in\Dclosure[\modord{d}] v_i^{[r_i]}$ for each $i\in[1,n]$,
  \item  $u_i\modord{d}\hat{u}_i\in\Dclosure[\modord{d}] \lambda^{r_i}(v_i)^*u_iv_{i+1}^*$ for $i\in[1,n-1]$,
  \item $u_0\modord{d}\hat{u}_0\in\Dclosure[\modord{d}] u_0v_1^*$, and
  \item $u_n\modord{d}\hat{u}_n\in\Dclosure[\modord{d}] \lambda^{r_i}(v_n)^* u_n$.
  \end{enumerate}

  We claim that for each $i\in[0,n]$, there are words $x_i,y_i$ so that
  \begin{enumerate}
    \item for each $i\in [1,n-1]$ for which $u_i$ is non-empty, $\hat{u}_i=x_iu_iy_i$ with $x_i\in \Dclosure[\modord{d}]\lambda^{r_i}(v_i)^*$, $y_i\in \Dclosure[\modord{d}]v_{i+1}^*$,
    \item $\hat{u}_0=u_0y_0$ and $y_0\in\Dclosure[\modord{d}] v_1^*$,
    \item $\hat{u}_n=x_nu_n$ and $x_0\in\Dclosure[\modord{d}] \lambda^{r_i}(v_{n})^*$.
  \end{enumerate}
  Note that is establishes the \lcnamecref{reg:association-ext}: We can then
  again modify the decomposition as follows.  We move $y_0$ from
  $\hat{u}_0$ to $\hat{v}_1$ and we move $x_n$ from $\hat{u}_n$ to
  $\hat{v}_n$. Moreover, for each non-empty $u_i$, we move $x_i$ from
  $\hat{u}_i$ to $\hat{v}_i$ and we move $y_i$ from $\hat{u}_i$ to
  $\hat{v}_{i+1}$. Each $\hat{u}_i$ where $u_i$ is empty is left
  unchanged.  The resulting decomposition
  $w=\bar{u}_0\bar{v}_1\bar{u}_1\cdots \bar{v}_n\bar{u}_n$ is then as
  desired.
  
  First, note that if some $u_i$ is empty (whether $i\in[1,n-1]$ or
  $i\in\{0,n\}$), then we need not construct any $x_i$ and $y_i$. We
  show how to construct $x_i$ and $y_i$ for $i\in[1,n-1]$ where $u_i$
  is non-empty. The proof for $y_0$ and $x_n$ is then analogous.
  
  Recall that $u_i\modord{d}\hat{u}_i\in\Dclosure[\modord{d}]
  \lambda^{r_i}(v_i)^* u_i v_{i+1}^*$.  This means there is some
  $\ell$ so that $\hat{u}_i\modord{d}\lambda^{r_i}(v_i)^\ell u_i
  v_{i+1}^\ell$.  Consider the $d$-embedding $\alpha$ of $u_i$ into
  $\hat{u}_i$ and the $d$-embedding $\beta$ of $\hat{u}_i$ into
  $\lambda^{r_i}(v_i)^\ell u_iv_{i+1}^\ell$. The composition $\gamma$
  of $\alpha$ and $\beta$ is a $d$-embedding of $u_i$ into
  $\lambda^{r_i}(v_i)^\ell u_i v_{i+1}^\ell$.

  We now use the fact that our extended loop pattern is
  irreducible. The $d$-embedding $\gamma$ cannot send the left-most
  letter of $u_i$ to a position in $\lambda^{r_i}(v_i)^\ell
  u_iv_{i+1}^\ell$ left of $u_i$, because that would mean that this
  letter is contained in $\kappa_d(v_i)(r_i+1)$. Moreover, $\gamma$
  cannot send the right-most letter of $u_i$ to a position in
  $\lambda^{r_i}(v_i)^\ell u_iv_{i+1}^\ell$ to the right of $u_i$,
  because that would mean that this letter is contained in
  $\kappa_d(v_{i+1})(d)$. This implies that $\gamma$ sends $u_i$
  exactly to the factor $u_i$ of $\lambda^{r_i}(v_i)^\ell u_i
  v_{i+1}^\ell$.  Thus, $\hat{u}_i$ has a factor $u_i$ that is sent by
  $\beta$ to $u_i$ of $\lambda^{r_i}(v_i)^\ell u_i v_{i+1}^\ell$. Let
  $\hat{u}_i=x_iu_iy_i$ be the corresponding decomposition. Then
  $\beta$ has to map $x_i$ into $\lambda^{r_i}(v_i)^\ell$ and $y_i$
  into $v_{i+1}^\ell$. In particular, we have
  $x_i\in\Dclosure[\modord{d}] \lambda^{r_i}(v_i)^*$ and
  $y_i\in\Dclosure[\modord{d}] v_{i+1}^*$. This completes the proof of
  the claim and hence the \lcnamecref{reg:association-ext}.
\end{proof}

\subsection{Proof of \cref{pump-adherence}}

\begin{lem}\label{reg:pump-word-up}
  Let $\cA$ be an automaton with $\le m$ states and let $d$ be a
  multiple of $m^3!$. Moreover, let $\pi_d(v)\le m^2$ and let
  $u\in\Dclosure[\modord{d}] v^{[r]}$ be accepted by $\cA$ such that
  $|u|\ge m\cdot \pi_d(v)$.  Then there is a
  $u'\in\Dclosure[\modord{\ell\cdot d}] (v^\ell)^{[r']}$ in
  $\langof{\cA}$ such that $r'=r+(\ell-1)d$ and $|u'|=|u|+(\ell-1)d$.
\end{lem}
\begin{proof}
Since $|u|\ge m\cdot \pi_d(v)$,
$u$ begins with at least $|u|/\pi_d(v)\ge m$ factors of length $\pi_d(v)$.
Consider the run of $\cA$ on $u$. Since $\cA$ has at most $m$ states,
we can decompose $u=fgh$ such that $g$ is
a contiguous block of $k\le m$ factors of length $\pi_d(v)$ and $g$ is read on a cycle.
Since $|g|=k\cdot\pi_d(v)\le m^3$, $|g|$ divides $d$. Let $u'=fg^{1+(\ell-1)d/|g|}h$.
Then according to \cref{reg:pump-inside-ideal-ext}, we have $u'\in\Dclosure[\modord{d}] v^{[r]}$.
Therefore, $\kappa_d(u')\subseteq\kappa_d(v)$. This implies 
\[ \kappa_{\ell\cdot d}(u')\subseteq\kappa_{\ell\cdot d}(v)\subseteq\kappa_{\ell\cdot d}(v^\ell). \]
Moreover, note that $|u'|=|u|+(\ell-1)d\equiv r+(\ell-1)d=r'\pmod{\ell\cdot d}$
and thus $u'\in\Dclosure[\modord{d}] (v^\ell)^{[r']}$.
\end{proof}

\begin{lem}\label{reg:pad-power}
Let $\cA$ be an automaton with $\le m$ states and let $d$ be a multiple of $m^3!$.
Moreover, let
$v\in(\Sigma^d)^*$ with $\pi_d(v)\le m^2$.
If $u\in\langof{\cA}$ with $w\modord{d}u\in\Dclosure[\modord{d}] v^{[r]}$,
then there is a $u'\in\langof{\cA}$ with
$w\modord{\ell\cdot d}u'\in\Dclosure[\modord{\ell\cdot d}] (v^\ell)^{[r]}$
\end{lem}
\begin{proof}
  Since $w\modord{d}u$, we can write $u=u_0w_1u_1\cdots w_nu_n$, where
  $w=w_1\cdots w_n$ and $w_1,\ldots,w_n\in\Sigma$, and
  $u_i\in(\Sigma^d)^*$. Since $u\in\Dclosure[\modord{d}] v^{[r]}$, we
  have $\kappa_d(u)\subseteq\kappa_d(v)$ and hence
  $u_i\in\Dclosure[\modord{d}] \lambda^i(v)^*$ for $i\in[0,n]$.

  For each $i\in[0,n]$, we construct $u'_i$ as follows.
  Consider the run of $\cA$ on $u$ and suppose it reads $u_i$ from state $p_i$ to state $q_i$.
  \begin{itemize}
  \item If $u_i$ is empty, then $u'_i=u_i$. Note that then of course
    $u'_i\in\Dclosure[\modord{\ell\cdot d}]\lambda^i(v^\ell)^*$.
  \item If $u_i$ is non-empty, then we split $u_i$ in $|u_i|/d$
    factors of length $d$ and apply to each factor
    \cref{reg:pump-word-up}. This yields a word a word $u'_i$ such that
    $u'_i\in\Dclosure[\modord{\ell\cdot d}] (\lambda^i(v)^\ell)^*$ and
    so that $u'_i$ can be read from state $p_i$ to $q_i$. Moreover, we have $|u'_i|$ is a multiple of $\ell\cdot d$. Since
    $\lambda^i(v)^\ell=\lambda^i(v^\ell)$, we have
    $u'_i\in\Dclosure[\modord{\ell\cdot d}] \lambda^i(v^\ell)^*$. 
  \end{itemize}
  Therefore, the word $u'=u'_0w_1u'_1\cdots w_nu'_n$ is accepted by
  $\cA$, belongs to $\Dclosure[\modord{\ell\cdot d}] (v^\ell)^{[r]}$ and
satisfies $w\modord{\ell\cdot d} u'$.
\end{proof}

\begin{lem}\label{reg:pad-border}
  Let $\cA$ be an automaton with $\le m$ states and and let $d$ be a
  multiple of $2m^3!$.  Moreover, let $v_i\in(\Sigma^d)^*$ with
  $\pi_d(v_i)\le m^2$ for $i=1,2$.  If $u\in\langof{\cA}$ with
  $u\in\Dclosure[\modord{d}] v_1^*v_2^*$, then there is a
  $u'\in\langof{\cA}$ with $u'\in\Dclosure[\modord{\ell\cdot d}]
  (v_1^\ell)^*(v_2^\ell)^*$
\end{lem}
\begin{proof}
Let $K=\Dclosure[\modord{d}] v_1^*v_2^*$. Observe that $K$ consists precisely of
the words of the form $u=x_1\cdots x_p st y_1\cdots  y_q$, where for some $r\in[0,d-1]$,
\begin{itemize}
\item $x_i\in\Dclosure[\modord{d}] v_1^*$ and $x_i\in \Sigma^d$ for $i\in[1,p]$,
\item $s\in \Dclosure[\modord{d}] v_1^{[r]}$ and $|s|=r$,
\item $t\in\Dclosure[\modord{d}] \lambda^r(v_2)^{[d-r]}$, and $|t|=d-r$, and
\item $y_i\in\Dclosure[\modord{d}] v_2^*$ and $y_i\in\Sigma^d$ for $i\in[1,q]$.
\end{itemize}
On the one hand, all such words belong to $\Dclosure[\modord{d}]
v_1^*v_2^*$: The parts $s$ and $t$ arise when dropping length-$d$
blocks on the border between $v_1^*$ and $v_2^*$. On the other hand,
by induction on the number of deleted length-$d$ blocks, it follows
that any word in $\Dclosure[\modord{d}] v_1^*v_2^*$ is of that shape.

Since $|s|+|t|=d$, we have either $|s|\ge d/2$ or $|t|\ge d/2$. We
treat the case that $|s|\ge d/2$, the other case is analogous. 

We apply \cref{reg:pump-word-up} to each factor $x_1,\ldots,x_p,s,y_1,\ldots,y_q$.
Note that this is possible because each of these words has length either exactly $d$ or $\ge d/2$ and we have $\ge d/2\ge m^3!\ge m^3\ge m\cdot\pi_d(v_i)$ for $i=1,2$. 
This yields words $x'_1,\ldots,x'_p,s',y'_1,\ldots,y'_q$ such that
\begin{itemize}
\item  $x'_i\in\Dclosure[\modord{\ell\cdot d}] (v_1^\ell)^{[0]}$ for $i\in[1,p]$,
\item $s'\in \Dclosure[\modord{\ell\cdot d}] (v_1^\ell)^{[r']}$, where 
$r'=r+(\ell-1)d$,
\item  $y'_i\in\Dclosure[\modord{\ell\cdot d}] (v_2^\ell)^{[0]}$ for $i\in[1,q]$,
\item $\cA$ accepts $u'=x'_1\cdots x'_ps'ty'_1\cdots y'_q$.
\end{itemize}
Recall that $t\in \Dclosure[\modord{d}] \lambda^r(v_2)^{[d-r]}$. This means
$\kappa_d(t)\subseteq \kappa_d(\lambda^r(v_2))$ and hence
\[\kappa_d(t)\subseteq\kappa_d(\lambda^r(v_2)^\ell)=\kappa_d(\lambda^r(v_2^\ell)) \]
(recall that $\lambda^r(w)^\ell=\lambda^r(w^\ell)$ for every word $w$).
Therefore, we also have
\begin{equation} \kappa_{\ell\cdot d}(t)\subseteq\kappa_{\ell\cdot d}(\lambda^r(v_2^\ell)). \label{bordercutting}\end{equation}
Note that since $\pi_{\ell\cdot d}(v_2^\ell)=\pi_d(v_2)$ divides $d$, we can rotate the word $v_2^\ell$
by a multiple of $d$ without changing its image under $\kappa_{\ell\cdot d}(\cdot)$. Hence 
\[ \kappa_{\ell\cdot d}(\lambda^r(v_2^\ell))=\kappa_{\ell\cdot  d}(\lambda^{r+(\ell-1)d}(v_2^\ell)) \]

Together with \cref{bordercutting}, we may conclude that $t$ belongs
to $\Dclosure[\modord{\ell\cdot d}]
(\lambda^{r+(\ell-1)d}(v_2^\ell))^{[d-r]}$ according to \cref{reg:ideal-charact-ext}.
Therefore, the above characterization of $K$, adapted to
$\Dclosure[\modord{\ell\cdot d}] (v_1^\ell)^*(v_2^\ell)^*$, is
satisfied for the word $u'$ and hence
$u'\in\Dclosure[\modord{\ell\cdot d}] (v_1^\ell)^*(v_2^\ell)^*$.
\end{proof}

\begin{lem}\label{reg:pump-single-ideal}
  Let $\cA$ be a finite automaton with $\le m$ states and let $d$ be a
  multiple of $2m^3!$.  If $u_0v_1^{[r_1]}u_1\cdots v_n^{[r_n]}u_n$ is
  an irreducible extended loop pattern with $\pi_d(v_i)\le m^2$ such
  that its ideal belongs to $\Adh[\modord{d}]{\langof{\cA}}$, then for
  each $\ell\in\N$, the ideal 
  \begin{equation} \Dclosure[\modord{\ell\cdot d}]  u_0(v_1^\ell)^{[r_1]}u_1\cdots (v_n^\ell)^{[r_n]}u_n \label{pumped-ideal}\end{equation}
  belongs to $\Adh[\modord{\ell\cdot d}]{\langof{\cA}}$.
\end{lem}
\begin{proof}
  Since $u_0v_1^{[r_1]}u_1\cdots v_n^{[r_n]}u_n$ is irreducible and
  its ideal belongs to $\Adh[\modord{d}]{\langof{\cA}}$, we know from
  \cref{reg:association-ext} that the extended loop pattern is
  associated to $\langof{\cA}$.

  Let $I$ be the ideal in \cref{pumped-ideal}. Let $w_i$ be the length-$r_i$ prefix of $v_i$ for every $i\in[1,n]$.

  In order to show that $I$ belongs to $\Adh[\modord{\ell\cdot
    d}]{\langof{\cA}}$, we have to exhibit for each $k\in\N$ a word
  $w\in\langof{\cA}$ so that $u_0(v_1^\ell)^kw_1u_1\cdots
  (v_n^\ell)^kw_nu_n\modord{\ell\cdot d} w$ and $w\in I$.

  Let $k\in\N$. Because of association, there is a word
  $\bar{w}=\bar{u}_0\bar{v}_1\bar{u}_1\cdots
  \bar{v}_n\bar{u}_n\in\langof{\cA}$ such that for every $i\in[1,n]$,
  we have $v_i^{k\cdot \ell}w_i\modord{d} \bar{v}_i$ and
  $\bar{v}_i\in\Dclosure[\modord{d}] v_i^{[r_i]}$.  Moreover,
  $\bar{u}_0=u_0$, $\bar{u}_n=u_n$, and for each $i\in[1,n-1]$:
  \begin{itemize}
  \item If $u_i$ is not empty, then $\bar{u}_i=u_i$.
  \item If $u_i$ is empty, then $\bar{u}_i\in\Dclosure[\modord{d}] \lambda^{r_i}(v_i)^*v_{i+1}^*$.
  \end{itemize}
  Consider the run of $\cA$ on $\bar{w}$.  Using \cref{reg:pad-power},
  we can choose $\bar{v}'_i$ such that $v_i^{k\cdot
    \ell}w_i\modord{\ell\cdot d} \bar{v}'_i$ and
  $\bar{v}'_i\in\Dclosure[\modord{\ell\cdot\ell}] (v_i^\ell)^{[r_i]}$
  and so that it has a run parallel to $\bar{v}_i$ in $\cA$.
  Now consider $\bar{u}_i$ for $i\in[0,n]$.
  \begin{itemize}
    \item If $\bar{u}_i=u_i$, then choose $\bar{u}'_i=\bar{u}_i=u_i$.
    \item If $\bar{u}_i\ne u_i$, then $u_i$ is empty and
      $\bar{u}_i\in\Dclosure[\modord{d}]
      \lambda^{r_i}(v_i)^*v_{i+1}^*$. Then we use
      \cref{reg:pad-border} to choose $\bar{u}'_i$ such that
      $\bar{u}'_i$ has a run parallel to $\bar{u}_i$ in $\cA$ and
      $\bar{u}'_i\in\Dclosure[\modord{\ell\cdot d}]
      (\lambda^{r_i}(v_i^\ell))^*(v_{i+1}^\ell)^*$.
    \end{itemize}
    Now the resulting word
    $w'=\bar{u}'_0\bar{v}'_1\bar{u}'_1\cdots\bar{v}'_n\bar{u}'_n$ is
    accepted by the automaton $\cA$. This shows that the extended loop pattern
    $u_0(v_1^\ell)^{[r_1]}u_1\cdots (v_n^\ell)^{[r_n]}u_n$ is
    associated to $\langof{\cA}$ and hence the ideal $I$ belongs to
    $\Adh[\modord{\ell\cdot d}]{\langof{\cA}}$.
\end{proof}

\begin{proof}[Proof of \cref{pump-adherence}]
  Suppose there is an ideal in
  the adherence $\Adh[\modord{d}]{\langof{\cA_i}}$ for $i=1,2$.
  By \cref{reg:smallperiod}, there is a loop pattern $u_0v_1u_1\cdots
  v_nu_n$ for $\cM_d$ such that the ideal $I=\Dclosure[\modord{d}]
  u_0v_1^*u_1\cdots v_n^*u_n$ belongs to
  $\Adh[\modord{d}]{\langof{\cA_i}}$ for $i=1,2$ and $\pi_d(v_i)\le
  m^2$ for every $i\in[1,n]$. Using \cref{reg:make-extended-irreducible},
  we can construct an irreducible extended loop pattern
  \[ \bar{u}_0\bar{v}_1^{[r_1]}\bar{u}_1\cdots
  \bar{v}_n^{[r_n]}\bar{u}_n \] that induces $I$ and satisfies
  $\pi_d(\bar{v}_i)\le m^2$ for $i\in[1,n]$.  Now
  \cref{reg:pump-single-ideal} tells us that the ideal
  \[ \Dclosure[\modord{\ell\cdot d}]
  \bar{u}_0(\bar{v}_1^\ell)^{[r_1]}\bar{u}_1\cdots
  (\bar{v}_n^\ell)^{[r_n]}\bar{u}_n \] belongs to
  $\Adh[\modord{\ell\cdot d}]{\langof{\cA_i}}$ for $i=1,2$.
\end{proof}

\subsection{Proof of \cref{undecidability:reduction}}
\begin{proof}[Proof of \cref{undecidability:reduction}]
Recall that the \emph{Post Correspondence Problem} asks, given two
morphisms $\alpha,\beta\colon\Sigma^*\to\{1,2\}^*$, whether there
is a word $w\in\Sigma^+$ such that $\alpha(w)=\beta(w)$.  The standard
undecidability proof~\cite{Sipser2012introduction} constructs, given a Turing machine $M$,
morphisms $\alpha,\beta$ such that for $w\in\Sigma^*$,  any common
prefix of $\alpha(w)$ and $\beta(w)$ encodes a prefix of a computation
history of $M$. 
For our decidable set $D$, there exists a fixed terminating Turing
machine, so we can proceed as follows. Given a word $u\in D$, we can apply
this construction to compute in polynomial time morphisms
$\alpha,\beta\colon\Sigma^*\to\{1,2\}^*$ such that
\begin{enumerate}
\item[(i)] $u\in D$ iff there is a $w\in\Sigma^+$ with $\alpha(w)=\beta(w)$ and
\item[(ii)]\label{undecidability:length-bound} there exists $k\in\N$ so that for every $w\in\Sigma^*$, the words $\alpha(w)$ and $\beta(w)$
have no common prefix longer than $k$.
\end{enumerate}
We claim that $u\in D$ if and only if $L_{\alpha,\beta}$ and $E$ are
separable by $\BSS[\MOD]$.  Clearly, if $u\in D$, then the languages
$L_{\alpha,\beta}$ and $E$ intersect and cannot be separable.
Suppose $u\notin D$. Then (ii) implies that
$L_{\alpha,\beta}$ is included in
\vspace{-0.25cm}
\begin{multline*}
S_k=\{ a^rcb^s \mid r\not\equiv s\bmod 2^{k+1} \} \\
\cup \{a^rcb^s \mid \min(r,s)<2^{k+1}-1, r\ne s\}
\end{multline*}
because $x,y\in\{1,2\}^*$, $|x|,|y|>k$, have a common prefix of length
$>k$ iff $\nu(x)\equiv\nu(y)\bmod{2^{k+1}}$. Moreover, for
$x\in\{1,2\}^*$, we have $|x|\le k$ iff $\nu(x)<2^{k+1}-1$. Since $S_k$
is clearly definable in $\BSS[\MOD]$ and disjoint from $E$, this shows
that $L_{\alpha,\beta}$ and $E$ are separable by $\BSS[\MOD]$.
\end{proof}

\end{document}